\documentclass[3p]{elsarticle}

\usepackage{amssymb}
\journal{Journal of Differential Equations}

\usepackage[T1]{fontenc}
\usepackage{graphicx}
\usepackage{subfig}
\usepackage{amssymb}
\usepackage{epstopdf}
\usepackage{mathrsfs,amssymb}
\usepackage{graphicx}
\usepackage{amscd}
\usepackage{amsmath}
\usepackage{amssymb}
\usepackage{amsthm}
\usepackage{amsfonts}
\usepackage{fancyhdr}
\usepackage{amsxtra}
\usepackage{tikz}
\usepackage{pgfplots}
\usepackage{setspace}
\usepackage{enumitem}
\usepackage{xcolor}
\usepackage[active]{srcltx}

\newcommand{\comments}[1]{}
\def\mbf#1{\mathbf{#1}}

\def\cal#1{\mathcal{#1}}
 
\def\R{\right} 
\def\<{\langle} 
\def\>{\rangle}
\def\N{{\mathbb N}}
\def\Om{\Omega}
\def\ra{\rightarrow}

\newcommand{\be}{\begin{equation}}
\newcommand{\ee}{\end{equation}}
\newcommand{\bes}{\begin{equation*}}
\newcommand{\ees}{\end{equation*}}

\newcommand{\bea}{\begin{eqnarray}}
\newcommand{\eea}{\end{eqnarray}}
\newcommand{\beas}{\begin{eqnarray*}}
\newcommand{\eeas}{\end{eqnarray*}}

\newtheorem{theorem}{Theorem}
\newtheorem{lemma}{Lemma}
\newtheorem{corollary}{Corollary}
\newtheorem{proposition}{Proposition}

\newtheorem{remark}{Remark}

\def\bt{\begin{theorem}}
	\def\et{\end{theorem}}
\def\bp{\begin{proposition}}
	\def\ep{\end{proposition}}
\def\bl{\begin{lemma}}
	\def\el{\end{lemma}}
\def\bt{\begin{corollary}}
	\def\et{\end{corollary}}

\numberwithin{equation}{section}
\numberwithin{theorem}{section}
\numberwithin{proposition}{section}
\numberwithin{lemma}{section}
\numberwithin{corollary}{section}
\numberwithin{remark}{section}
\numberwithin{definition}{section}
\numberwithin{claim}{section}

\def\R{\mathbb{R}}
\begin{document}

\begin{frontmatter}

\title{Persistence time of solutions of the three-dimensional Navier-Stokes equations in 
	Sobolev-Gevrey classes}

\author{Animikh Biswas}
\address{Department of Mathematics \& Statistics, 
	University of Maryland Baltimore County \\
	Baltimore, MD 21250, USA}
\author{Joshua Hudson}
\address{Johns Hopkins University Applied Physics Laboratory\\
	Laurel, MD 20723, USA}
\author{Jing Tian\corref{mycorrespondingauthor}}
\address{Department of Mathematics, Towson University\\
	Towson, MD 21252, USA}
\cortext[mycorrespondingauthor]{Corresponding author}
\ead{jtian@towson.edu}
%
\begin{abstract}
	In this paper, we study existence times of strong solutions of the three-dimensional Navier-Stokes equations in  time-varying analytic Gevrey classes based on Sobolev spaces $H^s, s> \frac{1}{2}$. This complements the seminal work of Foias and Temam (1989) on $H^1$ based Gevrey classes, thus enabling us to improve estimates of the analyticity radius of solutions for certain classes of initial data. The main thrust of the paper consists in showing that the existence times in the much stronger Gevrey norms (i.e. the norms defining the analytic Gevrey classes which quantify the radius of real-analyticity of solutions) match the best known persistence times in Sobolev classes. 
	Additionally, as in the case of persistence times in the corresponding Sobolev classes, our existence times in Gevrey norms are optimal for 
	$\frac{1}{2} < s < \frac{5}{2}$.	
\end{abstract}

\end{frontmatter}
\section{Introduction}\label{sec:intro}
  We consider the incompressible Navier--Stokes equations (NSE) in a three-dimensional domain $\Omega=[0,L]^3$, equipped with the space-periodic boundary condition. The NSE, which are the governing equations of motion of a viscous, incompressible, Newtonian fluid, are given by
\begin{subequations}
	\begin{align*}
	&\frac{\partial {u}}{\partial t}-\nu \Delta {u}+\left({u} \cdot \nabla\right) {u}+\frac{1}{\rho}\nabla p=0,\\
	&\nabla \cdot {u}=0,\\
	&{u}(x, 0)={u}^0 (x),
	\end{align*}
\end{subequations}
where $x=(x_1, x_2, x_3) \in \Omega$, ${u}(x,t)=(u_1, u_2, u_3)$ is the unknown velocity of the fluid, $u^0=(u^0_1, u^0_2, u^0_3)$ is the initial velocity, $\nu>0$ is the kinematic viscosity of the fluid, $\rho$ is the density, and $p$ the unknown pressure. The incompressibility constraint is manifested in the divergence free condition $\nabla \cdot u=0$.

Recently, several authors \cite{benameur2010blow, cheskidov2016lower, CM2018, cortissoz2014lower, mccormick2016lower, robinson2012lower} have obtained  ``optimal'' existence times, and the associated blow-up rates, assuming they exist, for solutions of  the 3D NSE in Sobolev spaces $H^s, s>\frac{1}{2}$. 
In particular, in \cite{robinson2012lower}, by employing a scaling argument, Robinson, Sadowski and Silva established that the optimal existence time of a (strong) solution of the NSE in the whole space $\R^3$, for initial data
in $H^s, s>\frac{1}{2}$,\comments{short of the answer to the global regularity problem for the 3D NSE being in the affirmative,} is necessarily given by 
\be  \label{optimaltime}
T(u_0) \gtrsim \frac{1}{\|u_0\|_{H^s}^{\frac{4}{2s-1}}}.
\ee
The optimality refers to the fact that if one establishes an existence time which depends solely on $\|u_0\|_{H^s}$ which is better than \eqref{optimaltime}, i.e. has the form $T \gtrsim \frac{1}{\|u_0\|_{H^s}^{\gamma}}$ with $\gamma < \frac{4}{2s-1}$, then the NSE is globally well-posed in $H^s$. Observe that an existence time of the form \eqref{optimaltime} immediately yields the  blow-up rate 
\[
\|u(t)\|_{H^s} \gtrsim \frac{1}{(T_*-t)^{\frac{2s-1}{4}}}\, ,
\]
where $T_*<\infty$ is the putative blow-up time of $\|u(t)\|_{H^s}$.
It follows from the optimality of the existence time that this blow-up rate is also optimal \cite{robinson2012lower}.\comments{ i.e., if $\|u(t)\|_{H^s}$ blows-up at the rate
$\|\|u(t)\|_{H^s} \gtrsim \frac{1}{(T_*-t)^\gamma}$ with $\gamma > \frac{2s-1}{4}$ implies that $T_*$ is not a blow-up time} In the same work \cite{robinson2012lower}, the authors obtained the following existence/persistence  times  in the space $H^s$, namely, 
\be  \label{robinsontime}
T(u_0) \gtrsim \ 
\left\{
\begin{array}{l}
\frac{1}{\|u_0\|_{H^s}^{\frac{4}{2s-1}}}, \quad
\frac{1}{2} < s < \frac{5}{2},\quad  s \neq \frac{3}{2},\\
\frac{1}{\|u_0\|_{H^s}^{\frac{5}{2s}}}, \quad s > \frac{5}{2}.
\end{array}
\right.
\ee
Evidently, the existence time is optimal for $\frac{1}{2} < s < \frac{5}{2}, s \neq \frac{3}{2}$, while the existence time for $s > \frac{5}{2}$, though not optimal, is the best known to-date. 
The borderline cases, namely $s=\frac{3}{2}, s = \frac{5}{2}$, were subsequently considered by varying methods in \cite{cheskidov2016lower, CM2018, cortissoz2014lower, mccormick2016lower},
including Littlewodd-Paley decomposition and other harmonic analysis tools, the upshot being that the
optimal existence time $T \sim  \frac{1}{\|u_0\|_{H^s}^2}$ also holds for $s=\frac{3}{2}$, while the optimal existence time in $H^{5/2}$ is still open.
\comments{$\frac{1}{2} <s < \frac{5}{2}$ while the existence time in \eqref{robinsontime} for $s>\frac{5}{2}$, though not optimal, is the best known to date.}

The purpose of our present work is to investigate as to what extent the above mentioned existence/ persistence times (and the associated blow-up rates)   hold if one considers the evolution of the NSE in an  \emph{analytic Gevrey class}, equipped with the much stronger Gevrey norm which characterizes space analyticity,  with the goal of obtaining sharper lower bounds of the space-analyticity radius of the solutions. In fluid-dynamics, the space analyticity radius has an important physical interpretation: at this length scale, the viscous effects  and the (nonlinear) inertial effects  are roughly comparable, and  below this length scale, the Fourier spectrum decays exponentially  \cite{BJMT, dt, foias, hkr, hkr1, kuk}. In other words,  the space analyticity radius yields a Kolmogorov type \emph{dissipation length scale}  encountered in conventional turbulence theory.  The exponential decay property of high frequencies can be used to show that the finite dimensional Galerkin approximations  converge exponentially fast. For instance, in the case of the complex Ginzburg-Landau equation,  analyticity estimates are used in \cite{doel-t} to rigorously explain numerical observations that the solutions to this equation can be accurately represented by a very low-dimensional  Galerkin approximation, and that the ``linear'' Galerkin approximation performs just as well as the nonlinear one.  Furthermore, a surprising connection between possible abrupt change in analyticity radius (which is necessarily shown to be intermittent in \cite{BiF} if it occurs) and (inverse) energy cascade in 3D turbulence was found in \cite{BiF}. Other applications of analyticity radius  occur in establishing sharp temporal decay rates of solutions in higher Sobolev norms  \cite{biswas2012, oliver2000remark}, establishing geometric  regularity criteria for the Navier-Stokes and related equations and in measuring the  spatial complexity of fluid flow \cite{bradshaw, g, kuk1} and in the nodal parameterization of the attractor \cite{fr,fkr}.

In a seminal work, Foias and Temam \cite{foias1989gevrey} pioneered the use of Gevrey norms for estimating space analyticity radius for the Navier-Stokes equations which was subsequently used by many authors (see \cite{biswas2012, biswas2010navier, biswas2007existence, cao1999navier, ferrari1998gevrey}, and the references there in); closely related approaches can be found in \cite{BGK, gk, gk2}. In this work, Foias and Temam showed that starting with initial data in $H^1$, one can control the much stronger Gevrey norm of the solution up a time which is comparable to the optimal existence time of the strong solution in $H^1$.
The Gevrey class approach enables one to avoid cumbersome recursive estimation  of higher order derivatives and is known to yield optimal estimates of the analyticity radius  \cite{optimaltiti}. Other approaches to analyticity can be found in \cite{pav, masuda, miura2006} for the 3D NSE, \cite{Kalantarov} for the Navier-Stokes-Voight equation, \cite{dong1, dong} for the surface quasi-geostrophic equation, \cite{Ly} for the Porous medium equation, and \cite{ang} for certain nonlinear analytic semi-flows.

The (analytic) Gevrey norm of $u$ in the Sobolev space $H^s$, which we refer to as the \emph{Sobolev-Gevrey norm} here,  is defined by 
$\|e^{\alpha A^{\frac{1}{2}}}u\|_{H^s}$, where $A$ is the Stokes operator. We recall that the norms $\|u\|_{H^s}$ and $\|A^{s/2}u\|_{L^2}$ are equivalent for mean-zero, divergence-free vector fields \cite{cf}. 
In case $\|e^{\alpha A^{\frac{1}{2}}}u\|_{H^s} < \infty $, then $u$ is 
space-analytic and the uniform space analyticity radius of $u$ is bounded below by $\alpha$. We provide below a brief summary of, and comments on, our results.
\begin{enumerate}
\item Assume that the initial data $\|e^{\beta_0 A^{\frac{1}{2}}}u_0\|_{H^s} < \infty$
with $\beta_0 \ge 0$; $\beta_0=0$ corresponds to $u_0 \in H^s$. In this case, $\sup_{t \in [0,T]} 
\|e^{(\beta_0 +\beta t)A^{\frac{1}{2}}}u\|_{H^s} < \infty $  with $0\leq\beta \leq \frac{1}{2}$ for 
$T \sim \frac{1}{\|e^{\beta_0 A^{\frac{1}{2}}}u\|_{H^s}^{\frac{4}{2s-1}}}, \frac{1}{2} < s < \frac{3}{2}$ and $T \sim \frac{1}{\|e^{\beta_0 A^{\frac{1}{2}}}u\|_{H^s}^{2}}, s > \frac{3}{2}$ (see Theorem \ref{mainthem213}). The quantity 
$\beta t$ captures the gain in  analyticity due to the dissipation.
If we set $\beta=0$, then this gives a persistence time in the 
Gevrey class corresponding to $\beta_0$.
Note that the time of persistence of the solution in the Gevrey class in this result
coincides with the optimal time of existence \eqref{optimaltime} in the range $\frac{1}{2} < s < \frac{3}{2}$ but is far from optimal in the range
$\frac{3}{2} < s < \frac{5}{2}$ and is also smaller than the best known  existence time in Sobolev classes in case $s>\frac{5}{2}$ obtained in
\cite{robinson2012lower}. The case $s=1$ is precisely the classical result of Foias and Temam \cite{foias1989gevrey}, while this result for  $ \frac{1}{2} \le s < \frac{3}{2}$ was obtained using semigroup methods in \cite{BiSw, biswas2010navier}. We provide a proof of this result using energy technique,  mainly for completeness, but also to illustrate that one can as a consequence, adapt a technique from 
\cite{dt} to obtain an improved estimate of the analyticity radius, which is possible by considering the evolution of Gevrey norm in $H^s$ with $s>1$; see  Theorem \ref{analyticthm} and 
 Remark \ref{rmk:improvedanalyticity}. This provides one of our motivations for considering the evolution of the Gevrey norm in higher-order Sobolev spaces.
 \item Subsequently, in Theorem \ref{mainthem1} and Theorem \ref{mainthem2}, 
 we  improve the existence times in the Gevrey classes given in Theorem \ref{mainthem213}
 for $s$ in the range $s \ge \frac{3}{2},\ s \neq \frac{5}{2}$. The existence time in Gevrey classes obtained in Theorem \ref{mainthem2} for $\frac{3}{2} \le s <\frac{5}{2}$ is optimal,
 	i.e. coincides with \eqref{optimaltime} while the existence time obtained in Theorem \ref{mainthem1} for $s > \frac{5}{2}$ coincides with the best known existence time in Sobolev classes $H^s$ obtained in \cite{robinson2012lower}. In order to prove these results, we first obtain 
  refined commutator estimates of the nonlinear term in Lemma \ref{lem1}, Lemma \ref{inequ3w} and Lemma \ref{inequ4w} which exploit their respective orthogonality properties. These estimates are new to the best of our knowledge and are motivated by those in \cite{BiSQG, BMSSQG} obtained for the 
  surface quasi-geostrophic equations. Using these estimates,
 for initial data in $H^s, s \ge \frac{3}{2}, s \neq \frac{5}{2}$,
 we show that 
 $\sup_{t \in [0,T]}\|e^{\beta tA^{\frac{1}{2}}}u\|_{H^s} < \infty $
 where $T$ is given as in \eqref{optimaltime} in the said range of $s$ (for large data).
 It is worth mentioning that the differential inequalities for the evolution of the Gevrey norms that one obtains in these cases are non-autonomous; estimates of existence times of these given in  Lemma \ref{lemmaonzeta} and Lemma \ref{lemmaonX}, though elementary, may be new as well. Moreover, in Corollary \ref{corollary2}, we give an alternate proof for the persistence in the Sobolev class $H^s$ for the entire range $\frac{1}{2} < s < \frac{5}{2}$, thus unifying the results in \cite{robinson2012lower} and \cite{cheskidov2016lower, cortissoz2014lower, mccormick2016lower} and showing that the case $\frac{3}{2}$ is not a borderline in our approach. Furthermore, unlike in \cite{cheskidov2016lower, mccormick2016lower},
 our method is elementary and avoids any harmonic analysis machinery such as paraproducts and Littlewood-Paley decomposition.
 \item The study of blow up in Gevrey classes is of importance for the NSE as it was shown in \cite{BiF} that in certain situations, an abrupt change in analyticity radius (which in turn is measured by a Gevrey norm) is indicative of a strong inverse energy cascade.
 The persistence time in Theorem \ref{mainthem213} (set $\beta=0, \beta_0>0$) readily yields a blow-up rate provided there exists a time $T_*$ at which the analyticity radius possibily decreases from $\beta_0$ (and consequently $\|e^{\beta_0 A^{\frac{1}{2}}}u(t)\|_{H^s}$ blows up as $t$ approcahes $T_*$).
 This is substantially different from the blow-up of a sub-analytic Gevrey norm studied in \cite{benameur2014exponential,benameur2016blow}. 
 As we show in Corollary \ref{corollary1}, a blow-up of a sub-analytic Gevrey norm can only occur if the solution itself loses regularity; whether or not a solution loses regularity is precisely one of the millennium problems.  In other words, for a globally regular solution, persistence in a sub-analytic Gevrey class is guaranteed for all times. However, this is not necessarily the case for analytic Gevrey norms. For instance, it is not difficult to show that for forced NSE, there exists a body-force, and an initial data $u_0$ in a Gevrey class, such that the solution exists globally in $H^s$ while a Gevrey norm of the form $\|e^{(\beta_0 +\beta t)A^{\frac{1}{2}}}u\|_{H^s} < \infty $ blows up in finite time. This is due to restriction posed on the solution by the analyticity radius of the driving force. To the best of our knowledge however, an example of such a phenomenon in the unforced case is  unknown. Therefore it is of interest to determine the blow-up rate in Gevrey classes even for solutions that are globally regular. Although our Theorem \ref{mainthem213} provides a blow-up rate, this  may not be optimal for $s > \frac{3}{2}$. At the very least, the blow-up rate provided in \eqref{norm13} does not correspond to the best known rate in Sobolev classes e.g. in \cite{robinson2012lower}. We leave it as an 
 \emph{open problem} to determine whether these rates can be matched. Although we obtain existence time results for Gevrey classes that matches the existence times in 
 \cite{robinson2012lower, cheskidov2016lower, mccormick2016lower} in Theorem \ref{mainthem1} and Theorem \ref{mainthem2}, they are for time-varying Gevrey classes defined by $\|e^{(\beta t)A^{\frac{1}{2}}}u\|_{H^s}$, i.e. $\beta_0=0$, and therefore $u_0 \in H^s$. A similar result on existence time for $\beta_0>0$ will yield an improvement of the blow-up rate  in Gevrey classes. This is an open problem as well.
			
\end{enumerate}

\section{Main results}\label{sec:results}
Before describing our main results, we first establish some notation, concepts, and settings. Using the notation $\displaystyle \kappa_0=\frac{2\pi}{L}$, define the dimensionless length, time, velocity, and pressure variables
\begin{align*}
\tilde{x}=\kappa_0 x, \tilde{t}=\nu \kappa_0^2 t, \tilde{u}=\frac{u}{\nu \kappa_0}, \tilde{p}=\frac{p}{\rho \nu^2 \kappa_0^2}.
\end{align*}
Using this transformation, the NSE transform to
\begin{subequations}
	\begin{align*}
	&\frac{\partial {\tilde{u}}}{\partial \tilde{t}}-\tilde{\Delta} {\tilde{u}}+\left({\tilde{u}} \cdot \tilde{\nabla}\right) {\tilde{u}}+\tilde{\nabla} \tilde{p}=0,
	\\ &\tilde{\nabla} \cdot {\tilde{u}}=0,
	\\ &\tilde{u}(x, 0)=\tilde{u}^0 (x).
	\end{align*}
\end{subequations}
$\tilde{\Delta}$ and $\tilde{\nabla}$ denote the gradient and Laplacian operators with respect to the primed variables.
Henceforth, for simplicity, we assume that $\nu=1$, $L=2\pi$, $\rho=1$, and $\kappa_0=\frac{2\pi}{L}=1$. We have the dimensionless version of the NSE as
\begin{subequations}\label{nse}
	\begin{align}
	&\frac{\partial {{u}}}{\partial {t}}-{\Delta} {{u}}+\left({{u}} \cdot {\nabla}\right) {{u}}+{\nabla} p=0,\label{nse1}
	\\ &{\nabla} \cdot {{u}}=0,\label{nse2}
	\\ &{u}(x, 0)={u}^0 (x),\label{nse3}
	\end{align}
\end{subequations}
after dropping the tildes. 

Moreover, we will focus on $\Omega=[0, 2\pi]^3$, employ the Galilean invariance of the NSE, take $u$ to be mean free, i.e.,
$\displaystyle \int_{\Omega}u=0.$

In this paper, we are interested in investigating the existence times of strong solutions of the three-dimensional Navier-Stokes equations in time-varying analytic Gevrey classes based on Sobolev spaces $H^s, s> \frac{1}{2}$.  The results vary as the value of $s$ changes. 

\subsection{Functional analytic framework}
With $\Omega=[0, 2\pi]^3$, we denote by $\dot{L}^2(\Omega)$ the Hilbert space of all $L-$periodic functions from $\mathbb{R}^3$ to $\mathbb{R}^3$ that are square integrable on $\Omega$ with respect to the Lebesgue measure and mean free. The scalar product is taken to be the usual $L^2(\Omega)$ inner product
\begin{align*}
	(u,v)=\int_{\Omega} u(x)\cdot v(x) dx,
\end{align*}
and we denote
\begin{align*}
	\|u\|_{L^2}=(u,u)^{1/2}.
\end{align*}

The real separable Hilbert space $\mathit{H}$ is formed by the set of all $\mathbb{R}^3$-valued functions $u(x), x\in \mathbb{R}^3$, which has the Fourier expansion
\begin{align*}
u(x)=\sum_{k\in \mathbb{Z}^3\setminus \left \{(0,0,0)\right \}}\hat{u}(k)e^{ik\cdot x} \quad \mbox{(with}\, \hat{u}(0)=0\,\mbox{),}
\end{align*}\\
where the Fourier coefficients $\hat{u}(k)\in \mathbb{C}^3$, for all $k\in \mathbb{Z}^3\setminus \left \{(0,0,0)\right \}$, satisfy
\bes
\hat{u}_k=  \overline{\hat{u}_{-k}},\
k \cdot \hat{u}(k)=0,\ 
\mbox{for all}\  k\in \mathbb{Z}^3\setminus \left \{(0,0,0)\right \}
\ \mbox{and}\  \|u\|_{L^2}^2=
\sum_{k\in \mathbb{Z}^3\setminus \left \{(0,0,0)\right \}}|\hat{u}(k)|^2 < \infty.
\ees
For $s\geq 0$, the space $	\dot{H}^s (\Omega)$ is defined by
\begin{align*}
	\dot{H}^s (\Omega)=
	\left \{u\in H: u=\sum_{k\in \mathbb{Z}^3\setminus \left \{(0,0,0)\right \}}\hat{u}(k)e^{ik\cdot x},\  \|u\|_{H^s(\Om)}^2=\sum |k|^{2s} |\hat{u}_k|^2<\infty
	\right \}.
\end{align*}
\comments{
The norm in $\dot{H}^s (\Omega)$ is defined as
\begin{align*}
	\|u\|_{\dot{H}^s(\Omega)}=\left (\sum_{k\in \mathbb{Z}^3\setminus \left \{(0,0,0)\right \}} |k|^{2s}|\hat{u}_k|^2\right )^{1/2}.
\end{align*}
}
For simplicity, we denote  $\|\cdot\|_{\dot{H}^s(\Omega)}$ as 
$\|\cdot\|_{s}$.
For $s< 0$, the space $	\dot{H}^s (\Omega)$ is defined to be the dual of $\dot{H}^{|s|} (\Omega)$.
The $l^1-$type norm of the Fourier coefficients is given by
\begin{align*}
	\|u\|_{F^s(\Omega)}=\sum_{k\in \mathbb{Z}^3\setminus \left \{(0,0,0)\right \}} |k|^{s}|\hat{{u}}_k|.
\end{align*}
We write $\|u\|_{F}$ for $\|u\|_{F^0}$. It is easy to see that $F^s(\Omega)$ form an algebra under multiplication and $F^0(\Omega)$ is referred to as the Wiener algebra \cite{BJMT}.

\subsubsection{Gevrey class of functions} 
We say that a function $u \in C^\infty(\Omega)$ is in Gevrey class
$Gev(\alpha;\theta)$ if 
\be  \label{def:gevrey}
|\partial^\mbf{m} u(x)|\le 
M \left(\frac{\mbf m!}{\alpha^{|\mbf m|}}\right)^\theta\ \forall\  x \in \Omega,
\ee
where $\mbf m=(m_1,\cdots,m_n) \in \N^n$ is a multi-index,
$\mbf m !=m_1!\cdots m_n!$ and $|\mbf m |=\sum_{i=1}^nm_i$. The analytic Gevrey class corresponds to $\theta=1$, in which case, the function $u$ is real analytic with \emph{uniform analyticity radius} $\alpha$ for all 
$x \in \Om$. 
In case $0<\theta <1$,  the functions are called sub-anlytic. 
For a function $u \in H$, its Gevrey norm  is defined by
\begin{align*}	
\|u\|_{s, \alpha;\theta}=\|A^{\frac{s}{2}}e^{\alpha A^{\frac{\theta}{2}}}u\|_{L^2}=\|e^{\alpha A^{\frac{\theta}{2}}}u\|_s=\left (\sum_{k\in \mathbb{Z}^3\setminus \left \{(0,0,0)\right \}} |k|^{2s} e^{2\alpha|k|^\theta}|\hat{u}_k|^2\right )^{1/2},
\end{align*}
where $\alpha>0$. The connection between Gevrey class and Gevrey norm is
given by the fact that \eqref{def:gevrey} holds for all $x \in \Om$ if and only if $\|u\|_{s,\alpha;\theta} < \infty$ \cite{oliver2000remark,optimaltiti}. In case $\theta=1$, this is equivalent to the fact that $u$ is real analytic with uniform radius of real analyticity $\alpha$. We will denote  
\[
Gv(s,\alpha;\theta)=\left\{u \in H: \|u\|_{s,\alpha;\theta} < \infty \right\},
\]
and in case $\theta=1$, for simplicity, we will write $Gv(s,\alpha)$ instead of $Gv(s,\alpha;1)$ and we will denote $\|u\|_{s, \alpha;1}$ as 
$\|u\|_{s, \alpha}.$ Clearly, 
\[
Gv(s,\alpha)\subsetneq Gv(s,\alpha;\theta)
\subsetneq  \dot{H}^m(\Omega)\ \mbox{for all}\  0<\theta<1, s \in \R, 
m \in \R_+.
\]
If $u \in Gv(s,\alpha)$, then clearly
\[
|\hat{u}(k)| \le e^{-\alpha |k|}\|u\|_{s,\alpha},
\]
and therefore, the uniform analyticity radius $\alpha$ establishes a length scale below which the Fourier power spectrum decays exponentially which in turn relates it to the Kolmogorov decay length scale in turbulence theory \cite{BJMT, dt}.

The \emph{maximal analyticity radius} for a function $u \in H^s$ is defined by 
\[
\lambda_{max}(u)= \sup \{ \alpha \ge 0: \|u\|_{\alpha,s} < \infty\}.
\]
One can check easily that $\lambda_{max}(u)$ is independent of $s$.

\comments{
The Gevrey norm we consider here defines the analytic Gevrey class. We also consider the subanalytic Gevrey classes: $\displaystyle \left \{ u \:|\: \|e^{rA^{\frac{\theta}{2}}}u\|_s < \infty\right \}$, where $0<\theta<1$ and $r>0$.
}

\subsection{The Functional differential equation}
Let $\Pi$ be the orthogonal projection from $L^2$ onto the subset of $L^2$ consisting of those functions whose weak derivatives are divergence-free in the $L^2$ sense. $A$ is the Stokes operator, defined as
\begin{align}
\label{defA}
A=-\Pi \Delta.
\end{align}
$B$ is the bilinear form defined by 
\begin{align}
\label{defB}
B(u,u)=\Pi \left [({u} \cdot \nabla) {u}\right ].
\end{align}

Then, the functional form of the NSE can be written as
\begin{align}
\label{functional_form}
\frac{du}{dt}+\mathit{A}u+\mathit{B}(u,u)=0.
\end{align}
\subsection{Main results}
We will now present our main results. Here, we denote by $c$ all the dimensionless constants which are independent of $s$, while all the dimensionless constants which depend on $s$ are denoted by $c_s$. 
\begin{theorem}
	\label{mainthem213}
	Let $u$ be a strong solution of \eqref{nse} with initial condition $u^0\in Gv(s, \beta_0)(\Omega)$, for some $s>\frac{1}{2}$, $\beta_0\geq 0$, and $0\leq\beta \leq \frac{1}{2}$. 
	If $\|u^0\|_{s, \beta_0}\leq c_s$, then $\sup_{t\in[0,\infty)} \|u\|_{s, \beta_0+\beta t}<\infty$.
	
	If $\|u^0\|_{s, \beta_0}> c_s$, define
	\begin{align*}
	T^{\ast}=\sup \left\{T>0 \:|\: \sup_{t\in[0,T]}\|e^{(\beta_0+\beta t) A^{\frac{1}{2}}}u(t)\|_{s} < \infty \right\}.
	\end{align*}

We have
	\begin{align}\label{t13}
T^{\ast} \gtrsim \ 
\left\{
\begin{array}{l}
	\frac{1}{\|u^0\|^{\frac{4}{2s-1}}_{s, \beta_0}}, \quad
	\frac{1}{2} < s < \frac{3}{2}\\
	\frac{1}{\|u^0\|^{2}_{s, \beta_0}}, \quad s > \frac{3}{2}.
\end{array}
\right.
\end{align}

Moreover, 
	if $T^{\ast} < \infty$, $\|e^{(\beta_0+\beta t)A^{\frac{1}{2}}}u(t)\|_{s}$ will blow-up at the following rate
		\begin{align}\label{norm13}
			\|e^{(\beta_0+\beta t)A^{\frac{1}{2}}}u(t)\|_{{s}} \gtrsim \ 
		\left\{
		\begin{array}{l}
			\frac{1}{ (T^{\ast}-t)^{\frac{2{s}-1}{4}}}, \quad
			\frac{1}{2} < s < \frac{3}{2}\\
			\frac{1 }{ (T^{\ast}-t)^{\frac{1}{2}}}, \quad s > \frac{3}{2}.
		\end{array}
		\right.
	\end{align}
\end{theorem}
Proceeding as in \cite{BiF, dt}, we can optimize over the choice of $\beta$ to obtain a better lower estimate of the analyticity radius.
\begin{theorem}
	\label{analyticthm}
	Let $u$ be a strong solution of \eqref{nse} with initial condition $u^0\in Gv(s, \beta_0)(\Omega)$, for some $\frac{1}{2}<s<\frac{3}{2}$ and $\beta_0,\ \beta\geq 0$.
	When $t\in [0, t^{\ast})$
    \begin{align*}
    	\|u\|_{s, \beta_0+\beta t} \leq \frac{e^{\frac{\beta^2}{2}t}\|u(0)\|_{s, \beta_0}}{\left(1-\frac{2c_s}{\beta^2}\|u(0)\|_{s, \beta_0}^{\frac{4}{2s-1}}\left(e^{\frac{2\beta^2}{2s-1}t}-1\right)\right)^{\frac{2s-1}{4}}},
    \end{align*}
    where
    \begin{align*}
    t^{\ast}=\frac{2s-1}{2\beta^2}\log\left(1+\frac{\beta^2}{2c_s \|u(0)\|_{s, \beta_0}^{\frac{4}{2s-1}}}\right).
    \end{align*}   
    Moreover, for the optimal choice of  $\beta=\sqrt{2c_s}\|u(0)\|_{s, \beta_0}^{\frac{2}{2s-1}}\varsigma,$ with $\varsigma$ being the positive solution of 
    $-\frac{1}{2\varsigma^2}\log(1+\varsigma^2)+\frac{1}{1+\varsigma^2}=0$,
    a lower estimate of the analyticity radius is given by
    $$\lambda_{max}(u(t^*))\ge \beta_0+c_s\frac{1}{\|u(0)\|_{s, \beta_0}^{\frac{2}{2s-1}}}.$$
\end{theorem}
\begin{remark}  \label{rmk:improvedanalyticity}
\emph{
Let $u_0=\sum_{N \le |k|\le cN}\hat{u}(k)e^{i k \cdot x}, 1 \le c,$ with $\sum_k |\hat{u}(k)|^2 =1$ and observe that $\|u\|_s \sim N^s$.
Then by Theorem \ref{analyticthm} the lower estimate of the (gain in) analyticity radius is given by ${\displaystyle \frac{c_s}{N^{\frac{2s}{2s-1}}}}$. The lower estimate in \cite{dt} in this case is  
${\displaystyle \frac{c_1}{N^2}}$, which
corresponds to $s=1$.  Clearly, this  lower estimate improves in our case if one considers $1<s<\frac{3}{2}$. However, one cannot take the limit as $s \nearrow \frac{3}{2}$ in this estimate as $c_s \ra 0$.
}
\end{remark}
\begin{corollary}\label{corollary1}
	Let $u$ be a strong solution of \eqref{nse} with initial condition $u^0\in Gv(s,r_0;\theta)$, for some $s>\frac{1}{2}, r_0>0,$ and 
	$0<\theta <1$. Let
	\begin{align*}
	T^{\ddagger}=\sup \left\{T>0 \:|\: \sup_{t\in[0,T]} \|e^{r_0 A^{\frac{\theta}{2}}}u(t)\|_{s} < \infty \right\}.
	\end{align*}
If $T^{\ddagger} < \infty$, then as $t\nearrow T^{\ddagger}$,
$\lim_{t \nearrow T^{\ddagger}}\|u(t)\|_{s'}= \infty $ for any 
$s' > \frac{1}{2}$.
 Moreover, $\|u(t)\|_{Gv(s,r_0;\theta)}$ blows up at an exponential rate at $T^\ddagger$.
	\comments{
	\begin{align} \label{coro1}
	\|u(t)\|_{s}\geq \frac{c_s\|u^0\|_{L^2}^{1-\frac{2s}{5}}}{(T^{\ddagger}-t)^{\frac{2s}{5}}}.
	\end{align}
	}
\end{corollary}

\begin{theorem}
	\label{mainthem1}
	Let $u$ be a strong solution of the Navier--Stokes equations \eqref{nse} with initial condition $u^0\in \dot{H}^{s}(\Omega)$, for some $s>\frac{5}{2}$. 
	Let $0<\beta \leq \frac{1}{2}$, and define
	\begin{align*}
	T^{\ast}=\sup \big\{T>0 \:|\: \sup_{t\in[0,T]}\|e^{\beta t A^{\frac{1}{2}}}u(t)\|_{s} < \infty \big\}.
	\end{align*}
	(i) If $$\frac{\|u^0\|_s}{\|u^0\|_{L^2}}\geq c_s \beta^{-\frac{4s}{5}} \min \left\{1, \|u^0\|_{L^2}^{-\frac{2s}{5}}\right \},$$ then
	\begin{align*}
T^{\ast}>c_s \min \left\{1, \ \|u^0\|_{L^2}^{-1}\right \}\left(\frac{\|u^0\|_s}{\|u^0\|_{L^2}}\right)^{-\frac{5}{2s}}.
\end{align*}
	(ii) If $$\frac{\|u^0\|_s}{\|u^0\|_{L^2}}< c_s \beta^{-\frac{4s}{5}} \min \left\{1, \|u^0\|_{L^2}^{-\frac{2s}{5}}\right \},$$ then
	\begin{align*}
T^{\ast}>\min\left\{\tilde{Z}, \tilde{Z}^{2/5}\right\},
\end{align*}
where $\tilde{Z}=c_s \min \left\{1, \ \|u^0\|_{L^2}^{-1}\right \}\left(\frac{\|u^0\|_s}{\|u^0\|_{L^2}}\right)^{-\frac{5}{2s}}.$
\end{theorem}

\begin{theorem}
	\label{mainthem2}
	Let $u$ be a strong solution of \eqref{nse} with initial condition $u^0\in \dot{H}^{s}(\Omega)$, for some $\frac{3}{2}\leq s<\frac{5}{2}$. Let $0<\beta \leq \frac{1}{2}$, and define
	\begin{align*}
	T^{\ast}=\sup \left\{T>0 \:|\: \sup_{t\in[0,T]}\|e^{\beta t A^{\frac{1}{2}}}u(t)\|_{s} < \infty \right\}.
	\end{align*}
	(i) If $$\|{u}^0\|_{s}\geq \frac{c_{s} }{(\beta)^{\frac{2s-1}{2}}},$$
	then
	$$T^{\ast}>\frac{c_{{s}}}{\|u^0\|^{\frac{4}{2s-1}}_{{s}}}.$$
		(ii) If $$\|{u}^0\|_{s}< \frac{c_{s} }{(\beta)^{\frac{2s-1}{2}}},$$
	then
	$$T^{\ast}>\min\left\{{\cal N}, {\cal N}^{1/2}\right\},$$
	where $\displaystyle {\cal N}=\frac{c_{{s}}}{\|u^0\|^{\frac{4}{2s-1}}_{{s}}}.$
\end{theorem}
\begin{remark}
	\emph{The differential inequalities for the evolution of the Gevrey norms leading up to the proofs of Theorem \ref{mainthem1} and Theorem \ref{mainthem2} are non-autonomous and much more complicated than that of Theorem \ref{mainthem213}. Consequently, finding an optimal $\beta$ leading to an improved estimate of the analyticity radius as has been done in Theorem \ref{analyticthm} is difficult. Thus, it would be of interest to find an improved estimate of the analyticity radius for $s>\frac{3}{2}$ by optimizing over the choice of $\beta$.}
	
	\end{remark}
\begin{remark}
\emph{Following the technique presented in Theorem \ref{mainthem2}, we present  in the corollary below an alternate proof (i.e. different from the ones in \cite{cheskidov2016lower, CM2018, cortissoz2014lower,  mccormick2016lower, robinson2012lower})
of the existence time/blow-up rate
in spaces $H^s$ for the entire range $ \frac{1}{2} < s < \frac{5}{2}$
which in particular shows that the case $s=\frac{3}{2}$, which appears as a borderline case  in \cite{cheskidov2016lower, CM2018, cortissoz2014lower,  mccormick2016lower, robinson2012lower} is not really a borderline in our approach.}
\end{remark}
\begin{corollary}\label{corollary2}
	Let $u$ be a strong solution of \eqref{nse} with initial condition $u^0\in \dot{H}^{s}(\Omega)$, for some $s\in(\frac{1}{2},\frac{5}{2})$. Define
	\begin{align*}
	T^{\ddagger}=\sup \big\{T>0 \:|\: \sup_{t\in[0,T]} \|u(t)\|_{s} < \infty \big\}.
	\end{align*}
	 Then 
	\begin{align*}
	T^{\ddagger}>\frac{c_s}{\|u^0\|^{\frac{4}{2s-1}}_s}.
	\end{align*}
	
	Moreover, if $T^{\ast} < \infty$, then
	\begin{align} \label{coro2}
	\|u(t)\|_{{s}}>\frac{c_{s} }{ (T^{\ddagger}-t)^{\frac{2{s}-1}{4}}}.
	\end{align}
\end{corollary}

The rest of the paper is organized as follows. Section~\ref{sec:prelim} provides the background and setting for our analysis.
In Section~\ref{sec:velocity}, working on the velocity equation, we obtained new commutator estimates of the nonlinear term in Gevrey spaces. Using these estimates, in subsection~\ref{sec:general case}, the existence time and blow-up rates have been obtained for $\|u\|_{Gv(s, \beta_0+\beta t)}$ when $s>\frac{1}{2},\ s\neq \frac{3}{2}$. We have also obtained an improved estimate of the analyticity radius for $\|u\|_{Gv(s, \beta_0+\beta t)}$ when $ \frac{1}{2}<s<\frac{3}{2}$. In subsection~\ref{sec:s>5/2}, we improve the existence times in the Gevrey classes when $s>\frac{5}{2}$. In Section~\ref{sec:1/2<s<5/2}, working on the vorticity equaiton, we improve the existence times in the Gevrey classes when $\frac{3}{2}\leq s<\frac{5}{2}$. Section~\ref{sec:appe} is the Appendix which includes several proofs of several requisite lemmas\& propositions.

\section{Preliminaries}\label{sec:prelim}

We recall the definition of strong solutions from \cite{temam1995navier}.\\
Let $\displaystyle V=\left \{u\in H^1_{loc}(\Omega), \text{u is periodic, and}\ \nabla \cdot {u}=0\ \text{in}\ \Omega\right \}$ and $u_0\in V$, $u$ is a strong solution of NSE if it solves the variational formulation of (\ref{nse1})-(\ref{nse3}) as in \cite{cf, temam1995navier}, and
$$u\in L^2(0, T; D(A))\cap L^{\infty}(0, T; V),$$
for $T>0$.
The following lemma will be used in this paper.
\begin{lemma} \cite{swanson2011gevrey}
	\label{lem-gev1-n}
	Let $1<p<\infty$, if $s_1,\ s_2 <\frac{n}{p'},\ s_1+s_2 \geq 0$, and $s_1+s_2 > \frac{n}{p'}-\frac{n}{p}$, then
	\begin{align}
	\label{gevsplit-n}
	\|u \ast v\|_{s_1+s_2-\frac{n}{p'},p}\leq C_{s_1, s_2, n, p}\|u\|_{s_1, p} \|v\|_{s_2, p},
	\end{align}
	for all $u\in V_{s_1, p}$ and $v\in V_{s_2, p}$.
\end{lemma}
In our current setting, we have $n=3,\ p'=2,\ p=2.$ Since we mainly work in the Gevrey spaces, we will need another version of the above lemma.
\begin{lemma}
	\label{lem-gev1}
	In three dimensional spaces, for $s_1,\ s_2<\frac{3}{2}$ and $s_1+s_2> 0$, $u=e^{\alpha A^{\frac{1}{2}}}u_1\in \dot{H}^{s_1}$ and $v=e^{\alpha A^{\frac{1}{2}}}v_1\in \dot{H}^{s_2}$, we have
	\begin{align}
	\label{gevsplit}
	\|u_1 \ast v_1\|_{s_1+s_2-\frac{3}{2},\alpha}\leq \|u \ast v\|_{s_1+s_2-\frac{3}{2}}\leq C_{s_1,s_2}\|u_1\|_{s_1,\alpha} \|v_1\|_{s_2,\alpha}.
	\end{align}
\end{lemma}
\begin{lemma} \cite{mccormick2016lower} \label{lem4n}
	If $\displaystyle \dot{X}\leq c X^{1+\gamma}$ and $X(t)\rightarrow \infty$ as $t\rightarrow T$, then
	\begin{align*}
	X(t)\geq \left(\frac{1}{\gamma c(T-t)}\right)^{1/\gamma}.
	\end{align*}
\end{lemma}
\begin{lemma} \cite{robinson2012lower} \label{lem3}
	If $0\leq s_1<3/2+r<s_2$ and $u\in \dot{H}^{s_1}\cap\dot{H}^{s_2}$,  then $u\in F^r$ and
	\begin{align}
	\label{est4}
	\|u\|_{F^r}\leq c\|u\|_{s_1} ^{(s_2-r-3/2)/(s_2-s_1)}\|u\|_{s_2} ^{(3/2+r-s_1)/(s_2-s_1)}.
	\end{align}
\end{lemma}
\begin{lemma} \cite{robinson2012lower} \label{lemrobinson}
	Suppose that the local existence time in $\dot{H}^s (\mathbb{R}^3)$ depends on the norm in $\dot{H}^s (\mathbb{R}^3)$, with
	\begin{align*}
	T_{s} (u_0)\geq \frac{c'_s}{\|u_0\|_{{H}^s (\mathbb{R}^3)}}.
	\end{align*}
	Then
	\begin{align*}
	T_{s} (u_0)\geq {c_s}{\|u_0\|^{(5-2s)/2s}_{L^2 (\mathbb{R}^3)}}{\|u_0\|^{-5/2s}_{\dot{H}^s (\mathbb{R}^3)}}.
	\end{align*}
	In case the solution blows up at time $T<\infty$ then
	\begin{align*}
	\|u(T-t)\|_{\dot{H}^s (\mathbb{R}^3)} \geq c_s \|u(T-t)\|^{(5-2s)/5}_{L^2 (\mathbb{R}^3)} t^{-2s/5}.
	\end{align*}
\end{lemma}
We also need the following nonlinear generalization of the Gronwall inequality, which applies to the case of a nonlinear but positive vector field. For the proof, see Theorem~2.4 of \cite{online1}.
\begin{lemma} \cite{online1} \label{lemma:nonlinear_Gronwall}
Suppose that \(F(u,t)\) is a Lipschitz continous and monotonically increasing in $u$. Suppose that $u(t)$ is continuously differentiable, and $\displaystyle \frac{d}{dt} u(t) \leq F(u(t),t)$ for all \(t\in[0,T]\). Let $v$ be the solution of
	$\displaystyle \frac{d}{dt} v(t) = F(v(t),t),$ $v(0) = u(0),$
	and define 
	\[ T^* = \sup\left\{t>0 \:|\: \sup_{[0,t]} v(t) < \infty\right\}. \]
	Then \( u(t) \leq v(t) \) for all \( t\in \left[0,\min\{T,T^*\}\right]\).
\end{lemma}

In addition to the previous lemmas, we will also need to make use of several standard inequalities, which we present here for convenience. 

Young's inequality for products says that for nonnegative real numbers $a$ and $b$ and positive real numbers $p$ and $q$ satisfying $\frac{1}{p}+\frac{1}{q}=1$, we have:
$\displaystyle
ab\leq \frac{a^p}{p}+\frac{b^q}{q}.
$
We will frequently use Young's inequality with $p=q=2$:
$\displaystyle
ab\leq \frac{a^2}{2}+\frac{b^2}{2}.
$
Young's inequality with $\epsilon > 0$ will also be used: 
$\displaystyle
ab\leq \frac{a^2}{2\epsilon}+\frac{\epsilon b^2}{2}.
$

H\"older's inequality for sequences generalizes the Cauchy--Schwartz inequality. It states that for $p, q \in [1, \infty)$ satisfying $\frac{1}{p}+\frac{1}{q}\leq 1$
\begin{align*}
\sum^{\infty}_{k=1}|x_k y_k|\leq \left (\sum^{\infty}_{k=1}|x_k|^p\right )^{\frac{1}{p}}\left (\sum^{\infty}_{k=1}|y_k|^q\right )^{\frac{1}{q}}.
\end{align*}

The following energy estimate for the incompressible NSE (due to Leray) is essential, and allows us to bound the $L^2$ norm of any solution of \eqref{nse} by that of its initial data
\begin{align}
\label{energy1}
\|u(t)\|^2 _{L^2}+2 \int_{0}^{t}\|\nabla u(s)\|^2 _{L^2}ds\leq \|u^0\|^2 _{L^2}.
\end{align}

\section{Estimates on the velocity equation}\label{sec:velocity}
We start from the functional form (\ref{functional_form}) of the NSE
\begin{align*}
u_t+Au+B(u,u)=0.
\end{align*}
We can obtain the following estimates for the nonlinear term. The proofs of the following two lemmas which provide the main estimates of the nonlinear term are in the Appendix.
\begin{lemma}
	\label{lem1}
	(i) For $\forall s>0$, and $\forall u\in Gv(s+1, \alpha)\cap F^0$, we have
	\begin{align}
	\label{term329}
	\left|\left(B(u,u),A^{s}e^{2\alpha A^{\frac{1}{2}}}u\right)\right |\leq c_{s}\|e^{\alpha A^{\frac{1}{2}}}u\|_{F^0} \|u\|_{s, \alpha}
	\|u\|_{s+1, \alpha}.
	\end{align}
	(ii) For $\forall s\geq 1$, and $\forall u\in Gv(s+1, \alpha)\cap F^1$, we have
	\begin{align}
	\label{term32}
	\left|\left(B(u,u),A^{s}e^{2\alpha A^{\frac{1}{2}}}u\right)\right |\leq c_{s}\|e^{\alpha A^{\frac{1}{2}}}u\|_{F^1} \|u\|^2_{s, \alpha}
	+c_s \alpha \|e^{\alpha A^{\frac{1}{2}}}u\|_{F^1}\|u\|_{s+1, \alpha} \|u\|_{s, \alpha},
	\end{align}
	and consequently, 
	\begin{align}
	\label{term33}
	\left|\left(B(u,u),A^{s}e^{2\alpha A^{\frac{1}{2}}}u\right)\right |\leq c_{s}\|e^{\alpha A^{\frac{1}{2}}}u\|_{F^1} \|u\|^2_{s, \alpha}
	+{c_s \alpha^2} \|e^{\alpha A^{\frac{1}{2}}}u\|^2_{F^1}\|u\|^2_{s, \alpha}+\frac{1}{2}\|u\|^2_{s+1, \alpha}.
	\end{align}
\end{lemma}

We also obtain the following estimates on $\|e^{\alpha A^{\frac{1}{2}}}u\|_{L^2}$.
\begin{lemma}
	\label{lem33}
	For all $s>0$ and for all $u\in Gv(s,  \alpha )\cap L^2$,
	$$
	\|e^{ \alpha A^{\frac{1}{2}}}u\|_{L^2}\leq \sqrt{e}\|u\|_{L^2}+(2 \alpha )^s \|u\|_{s,  \alpha }.
	$$
\end{lemma}
\subsection{Existence time for $\|u\|_{Gv(s, \beta_0+\beta t)}$ when $s>\frac{1}{2},\ s\neq \frac{3}{2}$}\label{sec:general case}
In the proofs below, we follow the customary practice of  providing \emph{a priori} estimates which can be  rigorously justified by first obtaining these estimates for the finite dimensional Galerkin system, the solutions to which exist for all times, and then passing to the limit.
\begin{lemma}
	\label{inequ1a9}
	When $s>0$, $\beta_0,\ \beta \geq 0$, the solution, $u$, of \eqref{nse} with initial data $u^0\in Gv(s, \beta_0)$ satisfies the following differential inequality
\begin{align}
	\label{velnew1}
		&\frac{1}{2}\frac{d }{d t}\|u\|^2_{s, \beta_0+\beta t}-\beta \|A^{\frac{1}{4}}e^{(\beta_0+\beta t)A^{\frac{1}{2}}}u\|^2_s+ \|u\|^2_{s+1, \beta_0+\beta t}\\
	&\leq c_{s}\|e^{(\beta_0+\beta t)A^{\frac{1}{2}}}u\|_{F^0} \|u\|_{s, \beta_0+\beta t}
	\|u\|_{s+1, \beta_0+\beta t} \nonumber.
\end{align}
\end{lemma}
\begin{proof}
		Starting from the functional form of the NSE
	\begin{align*}
	u_t+Au+B(u,u)=0,
	\end{align*}
	and taking inner product with $A^{s}e^{2(\beta_0+\beta t)A^{\frac{1}{2}}}u$, we have
	\begin{align}
	\label{innerpro}
	\left (\frac{d u}{d t},A^{s}e^{2(\beta_0+\beta t)A^{\frac{1}{2}}}u\right )+\left (Au,A^{s}e^{2(\beta_0+\beta t)A^{\frac{1}{2}}}u\right )+\left (B(u,u),A^{s}e^{2(\beta_0+\beta t)A^{\frac{1}{2}}}u\right )=0.
	\end{align}
	
	We can explore (\ref{innerpro}) term by term. For the first term,
		\begin{align}
	\label{term1}
	\left (\frac{d u}{d t},A^{s}e^{2(\beta_0+\beta t)A^{\frac{1}{2}}}u \right )&=\frac{1}{2} \frac{d}{dt} \|A^{\frac{s}{2}}e^{(\beta_0+\beta t)A^{\frac{1}{2}}}u\|_{L^2}^2-\beta (A^{s+\frac{1}{2}}e^{2(\beta_0+\beta t)A^{\frac{1}{2}}}u,u) \nonumber \\ 
	&=\frac{1}{2} \frac{d}{dt} \|A^{\frac{s}{2}}e^{(\beta_0+\beta t)A^{\frac{1}{2}}}u\|_{L^2}^2 - \beta \|A^{\frac{1}{4}}e^{(\beta_0+\beta t)A^{\frac{1}{2}}}u\|^2_{s}.
	\end{align}
	
	For the second term of (\ref{innerpro}), we can write it in terms of the Gevrey norm
	\begin{align}
	\label{term2}
	\left(Au,A^{s}e^{2(\beta_0+\beta t)A^{\frac{1}{2}}}u\right)=\left(A^{\frac{s}{2}}A^{\frac{1}{2}}e^{(\beta_0+\beta t)A^{\frac{1}{2}}}u,A^{\frac{s}{2}}A^{\frac{1}{2}}e^{(\beta_0+\beta t)A^{\frac{1}{2}}}u\right)=\|u\|^2_{s+1, \beta_0+\beta t}.
	\end{align}
	
	For the third term of (\ref{innerpro}), applying (\ref{term329}) with $\alpha=\beta_0+\beta t$, we have	
	\begin{align}
	\label{term3n}
	\left|\left(B(u,u),A^{s}e^{2(\beta_0+\beta t)A^{\frac{1}{2}}}u\right)\right |
	\leq c_{s}\|e^{(\beta_0+\beta t)A^{\frac{1}{2}}}u\|_{F^0} \|u\|_{s, \beta_0+\beta t}
	\|u\|_{s+1, \beta_0+\beta t}.
	\end{align}
	
	Substituting (\ref{term1}), (\ref{term2}), and (\ref{term3n}) into (\ref{innerpro}), we have (\ref{velnew1}).
\end{proof}

	\begin{proof}[\textbf{Proof of Theorem \ref{mainthem213}}]
		With $0 \leq \beta \leq \frac{1}{2}$, we have
		\begin{align*}
		\beta \|A^{\frac{1}{4}}e^{(\beta_0+\beta t)A^{\frac{1}{2}}}u\|^2_s\leq\frac{1}{2} \|e^{(\beta_0+\beta t)A^{\frac{1}{2}}}u\|^2_{s+1}.
		\end{align*}
		
When $s>\frac{1}{2}$, we have 
$$\|e^{(\beta_0+\beta t)A^{\frac{1}{2}}}u\|_{F^0}\leq c_s \|e^{(\beta_0+\beta t)A^{\frac{1}{2}}}u\|_{s+1}.$$

Therefore, (\ref{velnew1}) becomes
\begin{align*}
\frac{1}{2}\frac{d }{d t}\|u\|^2_{s, \beta_0+\beta t}+\frac{1}{2} \|u\|^2_{s+1, \beta_0+\beta t}
\leq c_{s}\|u\|_{s, \beta_0+\beta t}
\|u\|^2_{s+1, \beta_0+\beta t} \nonumber.
\end{align*}
If $\|u^0\|_{s, \beta_0}\leq \frac{1}{2c_s}$, then $\frac{d }{d t}\|u\|^2_{s, \beta_0+\beta t}\leq 0$, $\|u\|_{s, \beta_0+\beta t}$ remains bounded for all time and  $\|u\|_{s, \beta_0+\beta t}
\le \|u\|_{s, \beta_0}$.

Now suppose $\|u^0\|_{s, \beta_0}> \frac{1}{2c_s}$. Then we have the following cases.

(1) $\frac{1}{2}<s<\frac{3}{2}$: Applying Lemma \ref{lem3} on $e^{(\beta_0+\beta t)A^{\frac{1}{2}}}u$ with $r=0$, $s_1=s$, and $s_2=s+1$, we obtain
	\begin{align}
	\label{estoff0}
	\|e^{(\beta_0+\beta t)A^{\frac{1}{2}}}u\|_{F^0}\leq c \|e^{(\beta_0+\beta t)A^{\frac{1}{2}}}u\|^{s-1/2}_{s} \|e^{(\beta_0+\beta t)A^{\frac{1}{2}}}u\|^{3/2-s}_{s+1}.
	\end{align}	
	Therefore, (\ref{velnew1}) becomes
	\begin{align*}
	&\frac{1}{2}\frac{d }{d t}\|u\|^2_{s, \beta_0+\beta t}+\frac{1}{2}\|u\|^2_{s+1, \beta_0+\beta t}\leq c_{s}\|u\|^{s+1/2}_{s, \beta_0+\beta t}
	\|u\|^{5/2-s}_{s+1, \beta_0+\beta t}.
	\end{align*}	
	Apply Young's inequality and after simplification, we have
	\begin{align*}
	\frac{d }{d t}\|u\|_{s, \beta_0+\beta t}\leq c_{s} \|u\|^{\frac{2s+3}{2s-1}}_{s, \beta_0+\beta t}.
	\end{align*}
	Considering the blow up time $T^{\ast}$ of $\|u\|_{s, \beta_0+\beta t}$: if $T^{\ast} < \infty$, then, 
	as $t\nearrow T^{\ast}$,	
applying Lemma \ref{lem4n}, we have
\begin{align*}
\|e^{(\beta_0+\beta t)A^{\frac{1}{2}}}u(t)\|_{{s}}>\frac{c_{s} }{ (T^{\ast}-t)^{\frac{2{s}-1}{4}}}.
\end{align*}
This is equivalent to
\begin{align*}
T^{\ast}> \frac{c_s}{\|u^0\|^{\frac{4}{2s-1}}_{s, \beta_0}}.
\end{align*}

(2)$s>\frac{3}{2}$:
We have 
$$\|e^{(\beta_0+\beta t)A^{\frac{1}{2}}}u\|_{F^0}\leq c_s \|e^{(\beta_0+\beta t)A^{\frac{1}{2}}}u\|_{s},$$
therefore, (\ref{velnew1}) becomes
\begin{align*}
&\frac{1}{2}\frac{d }{d t}\|u\|^2_{s, \beta_0+\beta t}+\frac{1}{2}\|u\|^2_{s+1, \beta_0+\beta t}\leq c_{s}\|u\|^2_{s, \beta_0+\beta t}
\|u\|_{s+1, \beta_0+\beta t}.
\end{align*}
Apply Young's inequality and after simplification, we have
\begin{align*}
\frac{d }{d t}\|u\|_{s, \beta_0+\beta t}\leq c_{s} \|u\|^{3}_{s, \beta_0+\beta t}.
\end{align*}
Considering the blow up time $T^{\ast}$ of $\|u\|_{s, \beta_0+\beta t}$: if $T^{\ast} < \infty$, then, 
as $t\nearrow T^{\ast}$,	
applying Lemma \ref{lem4n}, we have
\begin{align*}
\|e^{(\beta_0+\beta t)A^{\frac{1}{2}}}u(t)\|_{{s}}>\frac{c_{s} }{ (T^{\ast}-t)^{\frac{1}{2}}}.
\end{align*}
This is equivalent to
\begin{align*}
T^{\ast}> \frac{c_s}{\|u^0\|^{2}_{s, \beta_0}}.
\end{align*}

%
\end{proof}
\textbf{Proof of Theorem \ref{analyticthm}}
\begin{proof}
We start from:
\begin{align*}
&\frac{1}{2}\frac{d }{d t}\|u\|^2_{s, \beta_0+\beta t}-\beta \|A^{\frac{1}{4}}e^{(\beta_0+\beta t)A^{\frac{1}{2}}}u\|^2_s+ \|u\|^2_{s+1, \beta_0+\beta t}\\
&\leq c_{s}\|e^{(\beta_0+\beta t)A^{\frac{1}{2}}}u\|_{F^0} \|u\|_{s, \beta_0+\beta t}
\|u\|_{s+1, \beta_0+\beta t} \nonumber.
\end{align*}

Applying (\ref{estoff0}), we have
\begin{align}
\label{equofu1}
&\frac{1}{2}\frac{d }{d t}\|u\|^2_{s, \beta_0+\beta t}-\beta \|A^{\frac{1}{4}}e^{(\beta_0+\beta t)A^{\frac{1}{2}}}u\|^2_s+ \|u\|^2_{s+1, \beta_0+\beta t}\\
&\leq c_{s}\|u\|^{s+1/2}_{s, \beta_0+\beta t}
\|u\|^{5/2-s}_{s+1, \beta_0+\beta t} \nonumber.
\end{align}

Since $\displaystyle \|A^{\frac{1}{4}}e^{(\beta_0+\beta t)A^{\frac{1}{2}}}u\|^2_s\leq \|u\|_{s, \beta_0+\beta t}\|u\|_{s+1, \beta_0+\beta t}$, applying Young's inequality, we have
\begin{align*}
\beta \|A^{\frac{1}{4}}e^{(\beta_0+\beta t)A^{\frac{1}{2}}}u\|^2_s\leq \frac{\beta^2}{2}\|u\|^2_{s, \beta_0+\beta t}+\frac{1}{2}\|u\|^2_{s+1, \beta_0+\beta t}.
\end{align*}
Moreover,
\begin{align*}
\|u\|^{s+1/2}_{s, \beta_0+\beta t}
\|u\|^{5/2-s}_{s+1, \beta_0+\beta t}\leq c_s \|u\|^{\frac{2(2s+1)}{2s-1}}_{s, \beta_0+\beta t}+\frac{1}{2}\|u\|^2_{s+1, \beta_0+\beta t}.
\end{align*}

Therefore, (\ref{equofu1}) becomes
\begin{align*}
\frac{1}{2}\frac{d }{d t}\|u\|^2_{s, \beta_0+\beta t}\leq c_{s}\|u\|^{\frac{2(2s+1)}{2s-1}}_{s, \beta_0+\beta t}+\frac{\beta^2}{2}\|u\|^2_{s, \beta_0+\beta t},
\end{align*}
or equivalently, since $\|u\|_{s, \beta_0+\beta t}\neq 0 $ for all $t>0$,
we have
\begin{align*}
\frac{d }{d t}\|u\|_{s, \beta_0+\beta t}\leq c_{s}\|u\|^{1+\frac{4}{2s-1}}_{s, \beta_0+\beta t}+\frac{\beta^2}{2}\|u\|_{s, \beta_0+\beta t}.
\end{align*}
\comments{
Since
\begin{align*}
e^{\frac{\beta^2}{2}t} \frac{d }{d t}(e^{-\frac{\beta^2}{2}t}\|u\|_{s, \beta_0+\beta t})&=e^{\frac{\beta^2}{2}t}\left(e^{-\frac{\beta^2}{2}t}\frac{d }{d t}\|u\|_{s, \beta_0+\beta t}-\frac{\beta^2}{2}e^{-\frac{\beta^2}{2}t}\|u\|_{s, \beta_0+\beta t}\right)\\
&=\frac{d }{d t}\|u\|_{s, \beta_0+\beta t}-\frac{\beta^2}{2}\|u\|_{s, \beta_0+\beta t}\\
&\leq c_{s}\|u\|^{1+\frac{4}{2s-1}}_{s, \beta_0+\beta t},
\end{align*}
}
Multiplying both sides by $e^{- \frac{\beta^2}{2}t}$, we have
\begin{align*}
\frac{d }{d t}(e^{-\frac{\beta^2}{2}t}\|u\|_{s, \beta_0+\beta t})\leq c_{s}e^{\frac{2\beta^2}{2s-1}t}(e^{-\frac{\beta^2}{2}t}\|u\|_{s, \beta_0+\beta t})^{1+\frac{4}{2s-1}}.
\end{align*}
\comments{
Denoting 
\begin{align*}
y(t)=e^{-\frac{\beta^2}{2}t}\|u\|_{s, \beta_0+\beta t},
\end{align*}
we obtain
\begin{align*}
\frac{d }{d t}y\leq c_{s}e^{\frac{2\beta^2}{2s-1}t}y^{1+\frac{4}{2s-1}}.
\end{align*}

Solving it, we have
\begin{align*}
		y(t)\leq \frac{y(0)}{\left(1-\frac{2c_s}{\beta^2}y(0)^{\frac{4}{2s-1}}\left(e^{\frac{2\beta^2}{2s-1}t}-1\right)\right)^{\frac{2s-1}{4}}}.
\end{align*}
}
Consequently,
\begin{align}
\label{unorm1}
\|u\|_{s, \beta_0+\beta t}\leq \frac{e^{\frac{\beta^2}{2}t}\|u(0)\|_{s, \beta_0}}{\left(1-\frac{2c_s}{\beta^2}\|u(0)\|_{s, \beta_0}^{\frac{4}{2s-1}}\left(e^{\frac{2\beta^2}{2s-1}t}-1\right)\right)^{\frac{2s-1}{4}}}.
\end{align}
This implies that $\|u\|_{s, \beta_0+\beta t}$ is finite on the interval $[0, t^{\ast})$, where
\begin{align*}
t^{\ast}=\frac{2s-1}{2\beta^2}\log\left(1+\frac{\beta^2}{2c_s \|u(0)\|_{s, \beta_0}^{\frac{4}{2s-1}}}\right).
\end{align*}
\comments{
Considering the Fourier modes, we have
\begin{align*}
|\hat{u}(k, t)|\leq \frac{e^{\left(\frac{\beta^2}{2}-\beta|k|\right)t}e^{-\beta_0 |k|}\|u(0)\|_{s, \beta_0}}{|k|^s \left(1-\frac{2c_s}{\beta^2}\|u(0)\|_{s, \beta_0}^{\frac{4}{2s-1}}\left(e^{\frac{2\beta^2}{2s-1}t}-1\right)\right)^{\frac{2s-1}{4}}}.
\end{align*}
Here, we can observe that the spectrum has an exponential decay length of $\beta_0+\beta t$ (when $t<t^{\ast}$).} 
Choosing $t=\frac{t^{\ast}}{2}$, then the associated analyticity radius 
$\lambda$ is 
$$\lambda=\beta_0+\frac{\beta t^{\ast}}{2}=\beta_0+\frac{2s-1}{4\beta}\log\left(1+\frac{\beta^2}{2c_s \|u(0)\|_{s, \beta_0}^{\frac{4}{2s-1}}}\right).$$
The value of $\beta$ that maximizes $\lambda$ is given by
$$\beta=\sqrt{2c_s}\|u(0)\|_{s, \beta_0}^{\frac{2}{2s-1}}\varsigma,$$
where $\varsigma$ is the positive solution of the equation
$$-\frac{1}{2\varsigma^2}\log(1+\varsigma^2)+\frac{1}{1+\varsigma^2}=0.$$
The corresponding analyticity radius at $t=\frac{t^{\ast}}{2}$ is
$$\lambda=\beta_0+c_s (2s-1)\frac{1}{\|u(0)\|_{s, \beta_0}^{\frac{2}{2s-1}}}.$$
\end{proof}
{\bf Proof of Corollary \ref{corollary1}.}
\begin{proof}
Assume that $T^\ddagger < \infty$. Then clearly 
\be  \label{gevblow}
\limsup_{t \nearrow T^\ddagger} \|u\|_{s,r_0; \theta}=\infty. 
\ee
Assume that $\lim_{t \nearrow T^\ddagger}\|u(t)\|_{s'} \neq \infty $, then, 
$\liminf_{t \nearrow T^\ddagger} \|u\|_{s'} < \infty $ and there exists
a sequence $\{t_j\}_{j=1}^\infty $ with $t_j \nearrow T^\ddagger$ and 
$\|u(t_j)\|_{s'} \le M < \infty $. From Theorem \ref{mainthem213}, 
it follows that there exists 
$T_M >0$ such that 
\be  \label{gevbd}
\sup_{t \in (0,T_M]}\|u(t_j+t)\|_{s',\beta t}=
K_M < \infty. 
\ee
Choose $t_{j_0}$ satisfying $t_{j_0} < T^\ddagger < t_{j_0}+T_M$. 
Let $2 \delta = T^\ddagger - t_{j_0}$. Then, due to \eqref{gevbd}, 
we have
\be \label{gevbd1}
\sup_{t \in [t_{j_0}+\delta, T^\ddagger)}\|u(t)\|_{s',\alpha_0}\le K_M,
\ee
where $\alpha_0=\beta \delta$. Observe now that for any $s,s',r_0,\alpha_0>0$ and $0<\theta<1$, $\forall  v \in Gv(s,\alpha_0)$, it's also in $Gv(s',r_0;\theta)$. We have 
\be  \label{andom}
\|v\|_{s',r_0;\theta} \le C_{s',s,r_0,\alpha_0}\|v\|_{s,\alpha_0}.
\ee
From inequalities \eqref{gevbd1} and \eqref{andom}, we obtain a contradiction to \eqref{gevblow}. Therefore, $\lim_{t \nearrow T^\ddagger}\|u(t)\|_{s'}=\infty.$

Consequently, due to \cite{benameur2016blow}, the subanalytic norm will blow up exponentially.

\end{proof}
\subsection{Existence time for $\|u\|_{Gv(s, \beta t)}$ when $s>\frac{5}{2}$}\label{sec:s>5/2}
We will need the following two lemmas to proceed.
\begin{lemma}\label{lemmaonzeta}
	Consider the differential equation
	\begin{align}
		\label{mainequ1}
\frac{d}{dt}  \zeta = c_{s} \gamma \zeta^{1+\frac{5}{2s}} +c_{s}  (\beta t)^{s-\frac{5}{2}}\zeta^{2}+{c_s (\beta t)^2} \gamma^{2}\zeta^{1+\frac{5}{s}} +{c_s (\beta t)^{2s-3}}  \zeta^{3},
\end{align}
with initial condition $\zeta(0)$, for $s>\frac{5}{2}$, $0<\beta\leq \frac{1}{2}$, and the local existence time $T_\zeta<\infty$. \\
When $\displaystyle {\zeta}({0})\geq c_s \beta^{-\frac{4s}{5}} \min \left\{\gamma^{\frac{2s}{2s-5}}, \gamma^{-\frac{2s}{5}}\right \} $, it holds that
\begin{align}
\label{equzeta1}
T_{\zeta}>\frac{c_s \min \left\{\gamma^{\frac{5}{2s-5}}, \ \gamma^{-1}\right \}}{\zeta(0)^{\frac{5}{2s}}}.
\end{align}
When $\displaystyle {\zeta}({0})< c_s \beta^{-\frac{4s}{5}} \min \left\{\gamma^{\frac{2s}{2s-5}}, \gamma^{-\frac{2s}{5}}\right \} $, it holds that
\begin{align}
\label{equzeta2n}
T_{\zeta}>\min\left\{Z, Z^{2/5}\right\},
\end{align}
where $\displaystyle Z=\frac{c_s \min \left\{\gamma^{\frac{5}{2s-5}}, \ \gamma^{-1}\right \}}{\zeta(0)^{\frac{5}{2s}}}.$
\end{lemma}
The proof of the above lemma is provided in the appendix. In the next lemma, we establish the crucial differential inequality associated to the evolution of the Gevrey norm.
\begin{lemma}
	\label{inequ1a}
	When $s>\frac{5}{2}$ and $0\leq\beta \leq \frac{1}{2}$, the solution, $u$, of \eqref{nse} with initial data $u^0\in \dot{H}^s$ satisfies the following differential inequality
	\begin{align}
	\label{newest4}
	\frac{d }{d t}\|u\|_{s, \beta t}&\leq c_{s}\|u\|^{1+\frac{5}{2s}}_{s,\beta t} \|u\|^{1-\frac{5}{2s}}_{L^2}+c_{s} (\beta t)^{s-\frac{5}{2}}\|u\|^{2}_{s,\beta t} \nonumber \\
	&+{c_s (\beta t)^2} \|u\|^{1+\frac{5}{s}}_{s, \beta t} \|u\|^{2-\frac{5}{s}}_{L^2}+{c_s (\beta t)^{2s-3}}  \|u\|^{3}_{s, \beta t}.
	\end{align}
\end{lemma}
\begin{proof}

%
	
Taking inner product with $A^{s}e^{2\beta tA^{\frac{1}{2}}}u$ of the NSE and applying (\ref{term33}) with $\alpha=\beta t$, we get
	\begin{align}
		\label{newest}
	&\frac{1}{2}\frac{d }{d t}\|u\|^2_{s, \beta t}-\beta \|A^{\frac{1}{4}}e^{\beta tA^{\frac{1}{2}}}u\|^2_s+ \|u\|^2_{s+1, \beta t}\\ \nonumber
	&\leq c_{s}\|e^{\beta tA^{\frac{1}{2}}}u\|_{F^1} \|u\|^2_{s, \beta t}+{c_s \beta^2}t^2 \|e^{\beta tA^{\frac{1}{2}}}u\|^2_{F^1}\|u\|^2_{s, \beta t}+\frac{1}{2}\|u\|^2_{s+1, \beta t}.
	\end{align}
	
%
	When $\beta \leq \frac{1}{2}$, applying the Poincar$\acute{e}$ inequality, we have
	$\displaystyle
	\beta \|A^{\frac{1}{4}}e^{\beta tA^{\frac{1}{2}}}u\|^2_s\leq\frac{1}{2} \|e^{\beta tA^{\frac{1}{2}}}u\|^2_{s+1}.
$
	
	Therefore, (\ref{newest}) yields
	\begin{align}
	\label{newest2}
	\frac{1}{2}\frac{d }{d t}\|u\|^2_{s, \beta t}\leq c_{s}\|e^{\beta tA^{\frac{1}{2}}}u\|_{F^1} \|u\|^2_{s, \beta t}+ {c_s \beta^2}t^2 \|e^{\beta tA^{\frac{1}{2}}}u\|^2_{F^1}\|u\|^2_{s, \beta t}.
	\end{align}
	
	Applying Lemma \ref{lem3}, in (\ref{est4}), and taking $r=1$, $s_1=0$, and $s_2=s$ in (\ref{est4}), for $\frac{5}{2}<s$ and $u\in L_2 \cap \dot{H}^s$, we obtain
	\begin{align*}
	\|u\|_{F^1}\leq c \|u\|^{\frac{s-\frac{5}{2}}{s}}_{L^2} \|u\|^{\frac{5}{2s}}_{s}.
	\end{align*}
	
	Replacing $u$ by $e^{\beta tA^{\frac{1}{2}}}u$, it follows that
	\begin{align}
	\label{f1}
	\|e^{\beta tA^{\frac{1}{2}}}u\|_{F^1}\leq c\|e^{\beta tA^{\frac{1}{2}}}u\|^{1-\frac{5}{2s}}_{L^2} \|e^{\beta tA^{\frac{1}{2}}}u\|^{\frac{5}{2s}}_{s}.
	\end{align}
	
	Squaring both sides of (\ref{f1}), we have
	\begin{align}
	\label{f1sq}
	\|e^{\beta tA^{\frac{1}{2}}}u\|^2_{F^1}\leq c \|e^{\beta tA^{\frac{1}{2}}}u\|^{2-\frac{5}{s}}_{L^2} \|e^{\beta tA^{\frac{1}{2}}}u\|^{\frac{5}{s}}_{s}.
	\end{align}
	
	Substituting (\ref{f1}) and (\ref{f1sq}) into (\ref{newest2}), we get
	\begin{align}
	\label{newest3}
	\frac{1}{2}\frac{d }{d t}\|u\|^2_{s, \beta t}\leq c_{s} \|e^{\beta tA^{\frac{1}{2}}}u\|^{1-\frac{5}{2s}}_{L^2} \|u\|^{2+\frac{5}{2s}}_{s, \beta t}+{c_s \beta^2}t^2 \|e^{\beta tA^{\frac{1}{2}}}u\|^{2-\frac{5}{s}}_{L^2}\|u\|^{2+\frac{5}{s}}_{s, \beta t}.
	\end{align}
	
	When $s>\frac{5}{2}$, $1-\frac{5}{2s}>0$, we have $(a+b)^{1-\frac{5}{2s}}\leq c_s (a^{1-\frac{5}{2s}}+b^{1-\frac{5}{2s}})$ for $a, b, c>0$. Therefore, applying Lemma \ref{lem33}, we have
	\begin{align}
	\label{el2}
	\|e^{\beta tA^{\frac{1}{2}}}u\|^{1-\frac{5}{2s}}_{L^2}\leq c_s \|u\|^{1-\frac{5}{2s}}_{L^2}+c_s (\beta t)^{s-\frac{5}{2}}\|e^{\beta tA^{\frac{1}{2}}}u\|^{1-\frac{5}{2s}}_{s}.
	\end{align}
	
	Similarly, since $2-\frac{5}{s}>0$, (i.e. $s>\frac{5}{2}$), we obtain
	\begin{align}
	\label{el2sq}
	\|e^{\beta tA^{\frac{1}{2}}}u\|^{2-\frac{5}{s}}_{L^2}\leq c_s \|u\|^{2-\frac{5}{s}}_{L^2}+c_s (\beta t)^{2s-5}\|e^{\beta tA^{\frac{1}{2}}}u\|^{2-\frac{5}{s}}_{s}.
	\end{align}
	
	Substituting (\ref{el2}) and (\ref{el2sq}) into (\ref{newest3}), 
	and after simplification, we have
	\begin{align*}
	\frac{1}{2}\frac{d }{d t}\|u\|^{2}_{s, \beta t}&\leq c_{s} \|u\|^{2+\frac{5}{2s}}_{s, \beta t} \|u\|^{1-\frac{5}{2s}}_{L^2}+c_{s}  (\beta t)^{s-\frac{5}{2}}\|u\|^{3}_{s, \beta t} \nonumber \\
	&+{c_s (\beta t)^2} \|u\|^{2+\frac{5}{s}}_{s, \beta t} \|u\|^{2-\frac{5}{s}}_{L^2}+{c_s (\beta t)^{2s-3}}  \|u\|^{4}_{s, \beta t},
	\end{align*}
	which leads to (\ref{newest4}).
\end{proof}

\begin{proof}[\textbf{Proof of Theorem \ref{mainthem1}}]
	
	From Lemma \ref{inequ1a}, we have
	\begin{align*}
		\frac{d }{d t}\|u\|_{s, \beta t}&\leq c_{s}\|u\|^{1+\frac{5}{2s}}_{s,\beta t} \|u\|^{1-\frac{5}{2s}}_{L^2}+c_{s} (\beta t)^{s-\frac{5}{2}}\|u\|^{2}_{s,\beta t} \nonumber \\
	&+{c_s (\beta t)^2} \|u\|^{1+\frac{5}{s}}_{s, \beta t} \|u\|^{2-\frac{5}{s}}_{L^2}+{c_s (\beta t)^{2s-3}}  \|u\|^{3}_{s, \beta t}.
	\end{align*}
	
	Let $\gamma = \|u^0\|_{L^2}^{1-\frac5{2s}}$. Using the energy estimate (\ref{energy1}), i.e., $\|u(t)\|_{L^2} \leq \|u^0\|_{L^2}$, we have
	\begin{align*}
	\frac{d }{d t}\|u\|_{s, \beta t} & \leq 
	c_{s} \|u\|^{1+\frac{5}{2s}}_{s, \beta t} \gamma 
	+ c_{s} (\beta t)^{s-\frac{5}{2}}\|u\|^{2}_{s, \beta t} \nonumber \\
	&+ {c_s (\beta t)^2} \|u\|^{1+\frac{5}{s}}_{s, \beta t} \gamma^{2}
	+ {c_s (\beta t)^{2s-3}}  \|u\|^{3}_{s, \beta t}.
	\end{align*}
	
	We will complete the proof using Lemma~\ref{lemma:nonlinear_Gronwall}. 
	Let \(\zeta(t)\) solve the differential equation
	\begin{align*}
	\frac{d}{dt}  \zeta = c_{s} \gamma \zeta^{1+\frac{5}{2s}} +c_{s}  (\beta t)^{s-\frac{5}{2}}\zeta^{2}+{c_s (\beta t)^2} \gamma^{2}\zeta^{1+\frac{5}{s}} +{c_s (\beta t)^{2s-3}}  \zeta^{3},
	\end{align*}
	with \(\zeta(0) = \zeta_0 = \|u^0\|_s\).
	
	
	Defining the local existence time of $\|u\|_{s, \beta t}$ to be 
		\[ T_u = \sup\left\{ t>0 \:|\: \sup_{r\in[0,t]}\|u(r)\|_{s, \beta r} < \infty \right\}, \] 
	and the local existence time of $\zeta$ to be 
	\[ T_\zeta = \sup\left\{ t>0 \:|\: \sup_{r\in[0,t]}|\zeta(r)| < \infty \right\}. \] 
	
	
	Then, using Lemma~\ref{lemma:nonlinear_Gronwall} we can say that $\zeta(t)\geq\|u(t)\|_{s, \beta t}$ for all $t\in\left[0,\min\{T_\zeta,T_{u}\}\right]$, and hence conclude \(T_{u}\geq T_\zeta\). Moreover, we assume $T_{u}<\infty$, so $T_{\zeta}<\infty$ (actually, we can see this easily from the differential equation of $\zeta$). To obtain a lower bound of $T_u$, we will now analyze $T_\zeta$.
	From Lemma \ref{lemmaonzeta}, when $0<\beta\leq \frac{1}{2}$, we have the following.\\
	Case (i):
	In case $$\|u^0\|_s={\zeta}({0})\geq c_s \beta^{-\frac{4s}{5}} \min \left\{\gamma^{\frac{2s}{2s-5}}, \gamma^{-\frac{2s}{5}}\right \}= c_s \beta^{-\frac{4s}{5}} \min \left\{\|u^0\|_{L^2}, \|u^0\|_{L^2}^{-\frac{2s-5}{5}}\right \},$$ 
	i.e., if $$\frac{\|u^0\|_s}{\|u^0\|_{L^2}}\geq c_s \beta^{-\frac{4s}{5}} \min \left\{1, \|u^0\|_{L^2}^{-\frac{2s}{5}}\right \},$$ 
	it holds that
	\begin{align*}
	T_{u}\geq T_{\zeta}>\frac{c_s \min \left\{\gamma^{\frac{5}{2s-5}}, \ \gamma^{-1}\right \}}{\zeta(0)^{\frac{5}{2s}}}=\frac{c_s \min \left\{\|u^0\|_{L^2}^{\frac{5}{2s}}, \ \|u^0\|_{L^2}^{\frac5{2s}-1}\right \}}{\|u^0\|_s^{\frac{5}{2s}}}=c_s \min \left\{1, \ \|u^0\|_{L^2}^{-1}\right \} \left(\frac{\|u^0\|_s}{\|u^0\|_{L^2}}\right)^{-\frac{5}{2s}}.
	\end{align*}
	
	
	Denoting the maximal time of existence of $\|e^{\beta tA^{\frac{1}{2}}}u\|_{s}$ to be $T^{\ast}$, we have
	\begin{align*}
	T^{\ast}>c_s \min \left\{1, \ \|u^0\|_{L^2}^{-1}\right \}\left(\frac{\|u^0\|_s}{\|u^0\|_{L^2}}\right)^{-\frac{5}{2s}}.
	\end{align*}
	
%
%

Case (ii):
In case $$\|u^0\|_s={\zeta}({0})< c_s \beta^{-\frac{4s}{5}} \min \left\{\gamma^{\frac{2s}{2s-5}}, \gamma^{-\frac{2s}{5}}\right \}= c_s \beta^{-\frac{4s}{5}} \min \left\{\|u^0\|_{L^2}, \|u^0\|_{L^2}^{-\frac{2s-5}{5}}\right \},$$ 
i.e.,
if $$\frac{\|u^0\|_s}{\|u^0\|_{L^2}}< c_s \beta^{-\frac{4s}{5}} \min \left\{1, \|u^0\|_{L^2}^{-\frac{2s}{5}}\right \},$$ 
it holds that
\begin{align}
\label{equzeta2}
T^{\ast}>\min\left\{\tilde{Z}, \tilde{Z}^{2/5}\right\},
\end{align}
where $\tilde{Z}=c_s \min \left\{1, \ \|u^0\|_{L^2}^{-1}\right \}\left(\frac{\|u^0\|_s}{\|u^0\|_{L^2}}\right)^{-\frac{5}{2s}}.$

%
%
\end{proof}
\section{Existence time for $\|u\|_{Gv(s, \beta t)}$  when $\frac{3}{2}\leq s<\frac{5}{2}$.}\label{sec:1/2<s<5/2}
It will be more convenient here to study the evolution in Gevrey classes using the vorticity equation instead of the velocity equation. As we will see below, this will enable us to avoid the borderline of the Sobolev embedding encountered in 
\cite{cheskidov2016lower, CM2018, cortissoz2014lower, mccormick2016lower, robinson2012lower}.
The  equation for evolution of vorticity ${\omega}=\nabla \times u$ is given by
\begin{align}
\label{nseW}
&{{\omega}}_t+ A{\omega}+B(u,{\omega})-B({\omega},u)=0,
\\ &{{\omega}}^0 (x)=\mathbf{{\omega}}(x, 0)=\nabla \times u^0 (x).
\end{align}
Here, the operators $A$ and $B$ are defined in (\ref{defA}) and (\ref{defB}), respectively.

Recall
$$\|{\omega}\|_{\tilde{s}, \alpha}=\|e^{\alpha A^{\frac{1}{2}}}{\omega}\|_{\tilde{s}}.$$

Since $\|{\omega}\|_{\tilde{s}, \alpha}=\|u\|_{\tilde{s}+1, \alpha}$, we are taking $s=\tilde{s}+1.$ We have the following estimates, proofs of which can be found in the Appendix.
\begin{lemma}
	\label{inequ3w}
	For $-\frac{1}{2}<\tilde{s}<\frac{3}{2}$ and ${\omega}\in Gv(\tilde{s}+1, \alpha)$, we have
	\begin{align}
	\label{term3w}
	\left|	\left (B({\omega}, u), A^{\tilde{s}}e^{2\alpha A^{\frac{1}{2}}}{\omega} \right )\right|\leq c_{\tilde{s}} \|{\omega}\|_{\tilde{s}, \alpha}^{\tilde{s}+\frac{3}{2}} \|{\omega}\|_{\tilde{s}+1, \alpha}^{\frac{3}{2}-\tilde{s}}. 
	\end{align}
\end{lemma}	

\begin{lemma}
	\label{inequ4w}
	For $-\frac{1}{2}< \tilde{s}<\frac{3}{2}$ and ${\omega}\in Gv(\tilde{s}+1, \alpha )$, we have
	\begin{align}
	\label{term4w}
	\left|	\left (B(u, {\omega}), A^{\tilde{s}}e^{2\alpha A^{\frac{1}{2}}}{\omega} \right ) \right|\leq c_s \|{\omega}\|^{\tilde{s}+\frac{3}{2}}_{\tilde{s}, \alpha } \|{\omega}\|^{\frac{3}{2}-\tilde{s}}_{\tilde{s}+1, \alpha }+c_s \alpha  \|{\omega}\|^{\tilde{s}+\frac{1}{2}}_{\tilde{s}, \alpha } \|{\omega}\|^{\frac{5}{2}-\tilde{s}}_{\tilde{s}+1, \alpha }. 
	\end{align}
\end{lemma}	
We will also need the following lemma concerning existence time of 
a non-autonomous differential equation to proceed the proof 
of which is provided in the appendix.
\begin{lemma}\label{lemmaonX}
	Let $X(t)$ satisfy
	\begin{align}
	\label{mainequ1X}
	\frac{d}{dt} X(t) =c_{\tilde{s}}X^{1+ \frac{4}{1+2\tilde{s}}}+c_{\tilde{s}}(\beta t)^{\frac{4}{2\tilde{s}-1}}X^{1+  \frac{4}{2\tilde{s}-1}},
	\end{align}
	with initial condition $X(0)$, $\frac{1}{2}<\tilde{s}<\frac{3}{2}$, $0<\beta\leq \frac{1}{2}$, and the local existence time $T_X<\infty$. \\
	When $\displaystyle X({0})\geq \frac{c_{\tilde{s}} }{(\beta)^{\frac{2\tilde{s}+1}{2}}} $, we have
	\begin{align}
	\label{equzeta1X3}
	T_{X}>\frac{c_{\tilde{s}}}{X(0)^{\frac{4}{1+2\tilde{s}}}}.
	\end{align}
	When $\displaystyle X({0})<\frac{c_{\tilde{s}} }{(\beta)^{\frac{2\tilde{s}+1}{2}}} $, we have
	\begin{align}
	\label{equzeta2X}
	T_{X}>\min\left\{Q, Q^{1/2}\right\},
	\end{align}
	where $\displaystyle Q=\frac{c_{\tilde{s}}}{X(0)^{\frac{4}{1+2\tilde{s}}}}$.
\end{lemma}

We can now study the existence time of the solutions of the NSE in the Gevrey spaces when $\frac{3}{2}\leq s<\frac{5}{2}$. First, we have the following Lemma.
\begin{lemma}
	\label{nvorticity}
	When $-\frac{1}{2}<\tilde{s}<\frac{3}{2}$, $\beta\geq 0$, we have the following differential inequality
	\begin{align}
	\label{innerpro2}
	\frac{1}{2} \frac{d}{dt} \|{\omega}\|_{\tilde{s}, \beta t}^2 - \beta \|{\omega}\|^2_{\tilde{s}+\frac{1}{2}, \beta t}+\|{\omega}\|_{\tilde{s}+1, \beta t}^2 \leq c_{\tilde{s}} \|{\omega}\|^{{\tilde{s}}+\frac{3}{2}}_{\tilde{s}, \beta t} \|{\omega}\|^{\frac{3}{2}-\tilde{s}}_{\tilde{s}+1, \beta t}+c_{\tilde{s}} \beta t \|{\omega}\|^{{\tilde{s}}+\frac{1}{2}}_{\tilde{s}, \beta t} \|{\omega}\|^{\frac{5}{2}-\tilde{s}}_{\tilde{s}+1, \beta t}. 
	\end{align}
\end{lemma}
\begin{proof}
	Taking the inner product of \eqref{nseW} with $A^{{\tilde{s}}}e^{2\beta tA^{\frac{1}{2}}}{\omega}$, we have
	\begin{align}
	\label{innerpro1}
	\left ({\omega}_t,A^{{\tilde{s}}}e^{2\beta tA^{\frac{1}{2}}}{\omega} \right )+\left ( A {\omega},A^{{\tilde{s}}}e^{2\beta tA^{\frac{1}{2}}}{\omega} \right )+\left (B(u, {\omega}), A^{{\tilde{s}}}e^{2\beta tA^{\frac{1}{2}}}{\omega} \right )-\left (B({\omega}, u), A^{{\tilde{s}}}e^{2\beta tA^{\frac{1}{2}}}{\omega} \right )=0.
	\end{align}
	
	Similar to the calculation in Section 4, we have
	\begin{align}
	\label{term1w}
	\left ({\omega}_t,A^{{\tilde{s}}}e^{2\beta tA^{\frac{1}{2}}}{\omega} \right )=\frac{1}{2} \frac{d}{dt} \|{\omega}\|_{\tilde{s}, \beta t}^2 - \beta \|{\omega}\|^2_{\tilde{s}+\frac{1}{2}, \beta t},
	\end{align}
	and
	\begin{align}
	\label{term2w}
\left ( A {\omega},A^{{\tilde{s}}}e^{2\beta tA^{\frac{1}{2}}}{\omega} \right )= \|{\omega}\|_{\tilde{s}+1, \beta t}^2.
	\end{align}
	
	Applying Lemma \ref{inequ4w} with $\alpha=\beta t$ and
	combining (\ref{term3w}), (\ref{term4w}), (\ref{term1w}), and (\ref{term2w}), the estimate of (\ref{innerpro1}) becomes
	\begin{align*}
	\frac{1}{2} \frac{d}{dt} \|{\omega}\|_{\tilde{s}, \beta t}^2 - \beta \|{\omega}\|^2_{\tilde{s}+\frac{1}{2}, \beta t}+\|{\omega}\|_{\tilde{s}+1, \beta t}^2 \leq c_{\tilde{s}} \|{\omega}\|^{{\tilde{s}}+\frac{3}{2}}_{\tilde{s}, \beta t} \|{\omega}\|^{\frac{3}{2}-\tilde{s}}_{\tilde{s}+1, \beta t}+c_{\tilde{s}} \beta t \|{\omega}\|^{{\tilde{s}}+\frac{1}{2}}_{\tilde{s}, \beta t} \|{\omega}\|^{\frac{5}{2}-\tilde{s}}_{\tilde{s}+1, \beta t}. 
	\end{align*}
\end{proof}

\begin{proof}[\textbf{Proof of Theorem \ref{mainthem2}}]
%
	For $\frac{3}{2}\leq s<\frac{5}{2}$, i.e., $\frac{1}{2}\leq \tilde{s}<\frac{3}{2}$
, we consider $\frac{1}{2}<\tilde{s}<\frac{3}{2}$ and $\tilde{s}=\frac{1}{2}$, separately.\\
	Case (1), $\frac{1}{2}<\tilde{s}<\frac{3}{2}$: Using Young's Inequality, we have
	\begin{align*}
	c_{\tilde{s}} \|{\omega}\|_{\tilde{s}, \beta t}^{\frac{3+2\tilde{s}}{2}} \|{\omega}\|_{\tilde{s}+1, \beta t}^{\frac{3-2\tilde{s}}{2}}\leq c_{\tilde{s}}\|{\omega}\|_{\tilde{s}, \beta t}^{2 \cdot \frac{3+2\tilde{s}}{1+2\tilde{s}}}+\frac{1}{4}  \|{\omega}\|_{\tilde{s}+1, \beta t}^2,
	\end{align*}
	and 
	\begin{align*}
	c_{\tilde{s}} \beta t \|{\omega}\|_{\tilde{s}, \beta t}^{\frac{1+2\tilde{s}}{2}} \|{\omega}\|_{\tilde{s}+1, \beta t}^{\frac{5-2\tilde{s}}{2}}\leq c_{\tilde{s}}(\beta t)^{\frac{4}{2\tilde{s}-1}}\|{\omega}\|_{\tilde{s}, \beta t}^{2 \cdot \frac{1+2\tilde{s}}{2\tilde{s}-1}}+\frac{1}{4}  \|{\omega}\|_{\tilde{s}+1, \beta t}^2.
	\end{align*}
	
	Taking $\beta \leq \frac{1}{2}$, appliying the Poincar$\acute{e}$ inequality, we have
	\begin{align*}
	\beta \|{\omega}\|^2_{\tilde{s}+\frac{1}{2}, \beta t}\leq \frac{1}{2} \|{\omega}\|_{\tilde{s}+1, \beta t}^2.
	\end{align*}
	
	Therefore, from (\ref{innerpro2}) we deduce
	\begin{align*}
	\frac{d}{dt} \|{\omega}\|_{\tilde{s}, \beta t}^2 \leq c_{\tilde{s}}\|{\omega}\|_{\tilde{s}, \beta t}^{2 \cdot \frac{3+2\tilde{s}}{1+2\tilde{s}}}+c_{\tilde{s}}(\beta t)^{\frac{4}{2\tilde{s}-1}}\|{\omega}\|_{\tilde{s}, \beta t}^{2 \cdot \frac{1+2\tilde{s}}{2\tilde{s}-1}}.
	\end{align*}
	
	After simplification, we have
	\begin{align}
	\label{innerpro3}
	\frac{d}{dt} \|{\omega}\|_{\tilde{s}, \beta t} \leq c_{\tilde{s}}\|{\omega}\|_{\tilde{s}, \beta t}^{1+ \frac{4}{1+2\tilde{s}}}+c_{\tilde{s}}(\beta t)^{\frac{4}{2\tilde{s}-1}}\|{\omega}\|_{\tilde{s}, \beta t}^{1+  \frac{4}{2\tilde{s}-1}}.
	\end{align}
	
	Let X(t) be the solution of the differential equation
	\begin{align}
	\label{innerpro4}
	\frac{d}{dt} X(t) =c_{\tilde{s}}X^{1+ \frac{4}{1+2\tilde{s}}}+c_{\tilde{s}}(\beta t)^{\frac{4}{2\tilde{s}-1}}X^{1+  \frac{4}{2\tilde{s}-1}}.
	\end{align}
	with $X_0 = X(0) = \|{\omega}^0\|_{\tilde{s}}$. Then, using Lemma~\ref{lemma:nonlinear_Gronwall}, we have $X(t)\geq\|\omega(t)\|_{s, \beta t}$ for all $t\in\left[0,\min\{T_X,T_{\omega}\}\right]$.
	Here, $T_X$ and $T_{\omega}$ are the local existence time of $X$ and $\|{\omega}\|_{\tilde{s}, \beta t}$, respectively. Moreover, we can conclude that \(T_{\omega}\geq T_X\), and we assume $T_{\omega}<\infty$, also, $T_{X}<\infty$.
	
	From Lemma \ref{lemmaonX}, when $0<\beta\leq \frac{1}{2}$, we get the following.\\
	Case (1a): 
	When $$ \|{u}^0\|_{s}=\|{\omega}^0\|_{\tilde{s}}=X({0})\geq \frac{c_{\tilde{s}} }{(\beta)^{\frac{2\tilde{s}+1}{2}}}=\frac{c_{s} }{(\beta)^{\frac{2s-1}{2}}}, $$ it holds that
	\begin{align}
	\label{equzeta1X1}
T_{\omega}\geq	T_{X}>\frac{c_{\tilde{s}}}{X(0)^{\frac{4}{1+2\tilde{s}}}}=\frac{c_{\tilde{s}}}{\|{\omega}^0\|_{\tilde{s}}^{\frac{4}{1+2\tilde{s}}}}.
	\end{align}
	
%
	Considering the existence time of $\|{\omega}\|_{\tilde{s}, \beta t}$ (i.e., $\|{u}\|_{s, \beta t}$): $T^{\ast}$, we have
	\begin{align}
	\label{blow4-n}
	T^{\ast}\geq T_{X}>\frac{c_{\tilde{s}}}{\|{\omega}^0\|^{\frac{4}{1+2\tilde{s}}}_{\tilde{s}}}=\frac{c_{{s}}}{\|u^0\|^{\frac{4}{2s-1}}_{{s},\beta t}}.
	\end{align}
	
%
%
%
	
	Case (1b): From Lemma \ref{lemmaonX},
	when $$ \|{u}^0\|_{s}=\|{\omega}^0\|_{\tilde{s}}=X({0})< \frac{c_{\tilde{s}} }{(\beta)^{\frac{2\tilde{s}+1}{2}}}=\frac{c_{s} }{(\beta)^{\frac{2s-1}{2}}}, $$ it follows that
	\begin{align}
	\label{equzeta1X2}
	T_{\omega}\geq	T_{X}>\min\left\{\frac{c_{\tilde{s}}}{X(0)^{\frac{4}{1+2\tilde{s}}}}, \frac{c_{\tilde{s}}}{X(0)^{\frac{2}{1+2\tilde{s}}}}\right\}=\min\left\{\frac{c_{\tilde{s}}}{\|{\omega}^0\|_{\tilde{s}}^{\frac{4}{1+2\tilde{s}}}}, \frac{c_{\tilde{s}}}{\|{\omega}^0\|_{\tilde{s}}^{\frac{2}{1+2\tilde{s}}}}\right\}.
	\end{align}

	In conclusion, for Case (1ii), we have
	\begin{align*}
	T^{\ast}>\min\left\{\frac{c_{{s}}}{\|u^0\|^{\frac{4}{2s-1}}_{{s},\beta t}}, \frac{c_{{s}}}{\|u^0\|^{\frac{2}{2s-1}}_{{s},\beta t}}\right\}.
	\end{align*}
	Case (2): when $\tilde{s}=\frac{1}{2}$, i.e. $s=\frac{3}{2}$, (\ref{innerpro2}) becomes
	\begin{align}
	\label{innerpro8}
	\frac{1}{2} \frac{d}{dt} \|{\omega}\|_{\frac{1}{2}, \beta t}^2 - \beta \|{\omega}\|^2_{1, \beta t}+ \|{\omega}\|_{\frac{3}{2}, \beta t}^2 \leq c_{\tilde{s}} \|{\omega}\|^{2}_{\frac{1}{2}, \beta t} \|{\omega}\|_{\frac{3}{2}, \beta t}+c_{\tilde{s}} \beta t \|{\omega}\|_{\frac{1}{2}, \beta t} \|{\omega}\|^{2}_{\frac{3}{2}, \beta t}. 
	\end{align}
	
	Comparing the terms on the right hand side of (\ref{innerpro8}), we can expect that there is a region (when $t$ and $\|{\omega}\|_{\frac{1}{2}, \beta t}$ are both small), the term $c_{\tilde{s}} \beta t \|{\omega}\|_{\frac{1}{2}, \beta t} \|{\omega}\|^{2}_{\frac{3}{2}, \beta t}$ can be absorbed by $ \|{\omega}\|_{\frac{3}{2}, \beta t}^2$.\\
	
	Let $\displaystyle \breve{c}=\frac{1}{4c_{\tilde{s}}\beta }$ and let $t^{\lozenge}$ as the solution of 
	$\displaystyle \|{\omega}\|_{\frac{1}{2}, \beta t}=\frac{\breve{c}}{t}.$ (If $\|{\omega}\|_{\frac{1}{2}, \beta t}$ does not blow up, then the Theorem holds. Assume $\|{\omega}\|_{\frac{1}{2}, \beta t}$ blows up, then such $t^{\lozenge}$ exists.) 
	
	When $0<t<t^{\lozenge}$, we have
	\begin{align*}
	\|{\omega}\|_{\frac{1}{2}, \beta t}<\frac{\breve{c}}{t}\Rightarrow \|{\omega}\|_{\frac{1}{2}, \beta t}<\frac{1}{4c_{\tilde{s}} \beta t},
	\end{align*}
	
	and consequently, from (\ref{innerpro8}), we obtain
	\begin{align*}
	\frac{1}{2} \frac{d}{dt} \|{\omega}\|_{\frac{1}{2}, \beta t}^2 - \beta \|{\omega}\|^2_{1, \beta t}+ \|{\omega}\|_{\frac{3}{2}, \beta t}^2 \leq c_{\tilde{s}} \|{\omega}\|^{2}_{\frac{1}{2}, \beta t} \|{\omega}\|_{\frac{3}{2}, \beta t}+\frac{1}{4} \|{\omega}\|^{2}_{\frac{3}{2}, \beta t}. 
	\end{align*}
	
	When $\beta \leq \frac{1}{2}$, apply Young's inequality to the above inequality and simplify it, we have
	\begin{align*}
	\frac{d}{dt} \|{\omega}\|_{\frac{1}{2}, \beta t}^2 <c_{\tilde{s}}\|{\omega}\|_{\frac{1}{2}, \beta t}^{4}\ \Rightarrow \frac{d}{dt} \|{\omega}\|_{\frac{1}{2}, \beta t} <c_{\tilde{s}}\|{\omega}\|_{\frac{1}{2}, \beta t}^{3}.
	\end{align*}
	
	Denoting $Y(t)=\|{\omega}\|_{\frac{1}{2}, \beta t}$, then we have
	\begin{align}
	\label{innerpro9}
	\frac{d}{dt} Y<c_{\tilde{s}}Y^{3}.
	\end{align}
	
	The local existence time of $Y$ is:
	$\displaystyle T_Y = \sup\left\{ t>0 \:|\: \sup_{r\in[0,t]}|Y(r)| < \infty \right\}. $
	
	We have $t^{\lozenge}<T_Y<\infty$, and when $0<t<t^{\lozenge}$, we compare $Y(t)$ with $\psi(t)$, where, $\psi(t)$ is the solution of
	\begin{align}
	\label{innerpro10}
	\frac{d}{dt} \psi=c_{\tilde{s}}{\psi}^{3},
	\end{align}
	with $\psi(0)=Y(0)$ with local existence time $T_{\psi}$, also $T_\psi<\infty$.\\
	
	Applying Lemma \ref{lemma:nonlinear_Gronwall} on (\ref{innerpro9}) and (\ref{innerpro10}), we have 
	\begin{align*}
	Y(t) \leq \psi(t),\ \text{for\ all}\ t\in\left [0, \min \left\{t^{\lozenge}, T_{X}, T_{\psi}\right\}\right].
	\end{align*}
	
Denoting the interception point of $\psi(t)$ with $\displaystyle \frac{\breve{c}}{t}$ as $t_{\psi}$, we have: 
$\displaystyle
	\psi(t_{\psi})=\frac{\breve{c}}{t_\psi}.
$
	Moreover, $t_{\psi}\leq t^{\lozenge}<T_{Y}$.
	
	Solving (\ref{innerpro10}), we have
	\begin{align}
	\label{tsol-nn}
	\psi(t)=(\psi(0)^{-2}-c_{\tilde{s}} t)^{-1/2}.
	\end{align}
	
	Therefore
	\begin{align*}
	(\psi(0)^{-2}-c_{\tilde{s}} t_{\psi})^{-1/2}=\frac{{\breve{c}}}{t_{\psi}}.
	\end{align*}
	
	After simplification, we obtain
	\begin{align*}
	\psi(0)^{-2}-c_{\tilde{s}} t_{\psi}={{\breve{c}}^{-2}}t_{\psi}^{2}\Rightarrow
	{{\breve{c}}^{-2}}t_{\psi}^{2}+c_{\tilde{s}} t_{\psi}=\psi(0)^{-2}.
	\end{align*}
	
	This is similar to the result in (\ref{equt-n}) with $\tilde{s}=\frac{1}{2}$. We follow similar procedure as in Case (1) and obtain the results on the existence time.
\end{proof} 

\begin{proof}[\textbf{Proof of the Corollary \ref{corollary2}}]
	From Lemma \ref{nvorticity}, 
	when $-\frac{1}{2}<\tilde{s}<\frac{3}{2}$, we have the following inequality
	\begin{align*}
	\frac{1}{2} \frac{d}{dt} \|{\omega}\|_{\tilde{s}, \beta t}^2 - \beta \|{\omega}\|^2_{\tilde{s}+\frac{1}{2}, \beta t}+\|{\omega}\|_{\tilde{s}+1, \beta t}^2 \leq c_{\tilde{s}} \|{\omega}\|^{{\tilde{s}}+\frac{3}{2}}_{\tilde{s}, \beta t} \|{\omega}\|^{\frac{3}{2}-\tilde{s}}_{\tilde{s}+1, \beta t}+c_{\tilde{s}} \beta t \|{\omega}\|^{{\tilde{s}}+\frac{1}{2}}_{\tilde{s}, \beta t} \|{\omega}\|^{\frac{5}{2}-\tilde{s}}_{\tilde{s}+1, \beta t}. 
	\end{align*}
	
	When we conside the Sobolev space, we have $\beta=0$. Applying Young's inequality on the above inequality and simplify it, we have
	\begin{align*}
	\frac{1}{2} \frac{d}{dt} \|{\omega}\|_{\tilde{s}}^2 \leq c_{\tilde{s}}\|{\omega}\|_{\tilde{s}}^{2 \cdot \frac{3+2\tilde{s}}{1+2\tilde{s}}}.
	\end{align*}
	
	Applying Lemma \ref{lem4n} and considering the existence time $T^{\ddagger}$ of $\|\mathbf{{\omega}}(t)\|_{\tilde{s}}$, we have
	\begin{align*}
	\|\mathbf{{\omega}}(T^{\ddagger}-t)\|_{\tilde{s}}\geq c_{\tilde{s}} t^{-\frac{1+2\tilde{s}}{4}}\Rightarrow \|\mathbf{{\omega}}(t)\|_{\tilde{s}}\geq c_{\tilde{s}} (T^{\ddagger}-t)^{-\frac{1+2\tilde{s}}{4}}.
	\end{align*}
	
	If we take $s=\tilde{s}+1$, so $\frac{1}{2}<s<\frac{5}{2}$, it follows that
	\begin{align*}
	\|u(t)\|_{s}=\|\mathbf{{\omega}}(t)\|_{\tilde{s}}\geq c_{\tilde{s}} (T^{\ddagger}-t)^{-\frac{1+2\tilde{s}}{4}}\geq{c_{s} } (T^{\ddagger}-t)^{-\frac{2s-1}{4}}.
	\end{align*}
	
	This is equivalent to 
	\begin{align*}
	 T^{\ddagger}>\frac{c_s}{\|u^0\|^{\frac{4}{2s-1}}_s}.
	\end{align*}
\end{proof}
\section{Appendix}\label{sec:appe}
\textbf{Proof of Lemma \ref{lem1}}.
\begin{proof}
	(i) Let us start by observing
	\begin{align*}
\left(B(u,u),A^{s}e^{2\alpha A^{\frac{1}{2}}}u\right)=\left(A^{s/2}e^{\alpha A^{\frac{1}{2}}}B(u,u),A^{s/2}e^{\alpha A^{\frac{1}{2}}}u\right).
	\end{align*}
	
		We just need to estimate the term $\displaystyle \|A^{s/2}e^{\alpha A^{\frac{1}{2}}}B(u,u)\|_{L^2}.$
	So we consider
	$\displaystyle I=\left(A^{s/2}e^{\alpha A^{\frac{1}{2}}}B(u,u),w\right),$
	for an arbitrary $w\in H$ with $\|w\|_{L^2}=1$. (In fact, we may take $w\in Gv(s, \alpha)$, and then pass to the limit in $H$. Accordingly, let $w\in Gv(s, \alpha)$ with $\|w\|_{L^2}=1$).
	
	\begin{align*}
	\left(A^{\frac{s}{2}}e^{\alpha A^{\frac{1}{2}}}B(u,u),w\right)&=\left(B(u,u),A^{\frac{s}{2}}e^{\alpha A^{\frac{1}{2}}}w\right)\\
	&=i \sum_{j,k} \left(j  \cdot \hat{u}_{k-j}\right)(\hat{u}_{j} \cdot \hat{w}_{-k})|k |^s e^{\alpha|k |}\\
	&=i \sum_{j,k} \left(k  \cdot \hat{u}_{k-j}\right)(\hat{u}_{j} \cdot \hat{w}_{-k})|k |^s e^{\alpha|k |},
	\end{align*}
	since $\hat{u}_{k-j} \cdot (k-j)=0$.
	
	The rest of the proof follows from the proof of the first inequality in Lemma 3.1 in \cite{robinson2012lower}. We also use the triangle inequality on the exponential function, namely,
	$$e^{\alpha|k|}\leq e^{\alpha|k-j|}e^{\alpha|j|}.$$
	
	(ii) Starting from the relation
	\begin{align*}
	\left(B(u,u),A^{s}e^{2\alpha A^{\frac{1}{2}}}u\right)=\left (A^{\frac{s}{2}}e^{\alpha A^{\frac{1}{2}}}B(u,u),A^{\frac{s}{2}}e^{\alpha A^{\frac{1}{2}}}u\right),
	\end{align*}
	note that since $\left(B(u,A^{\frac{s}{2}}e^{\alpha A^{\frac{1}{2}}}u),A^{\frac{s}{2}}e^{\alpha A^{\frac{1}{2}}}u\right)=0$, we have
	\begin{align}
	\label{estB}
	\left(B(u,u),A^{s}e^{2\alpha A^{\frac{1}{2}}}u\right)=\left(A^{\frac{s}{2}}e^{\alpha A^{\frac{1}{2}}}B(u,u)-B(u,A^{\frac{s}{2}}e^{\alpha A^{\frac{1}{2}}}u),A^{\frac{s}{2}}e^{\alpha A^{\frac{1}{2}}}u\right).
	\end{align}
	%
	
	We need to estimate
	$$\|A^{\frac{s}{2}}e^{\alpha A^{\frac{1}{2}}}B(u,u)-B(u,A^{\frac{s}{2}}e^{\alpha A^{\frac{1}{2}}}u)\|_{L^2}.$$
	Let us consider
	$$I=\left(A^{\frac{s}{2}}e^{\alpha A^{\frac{1}{2}}}B(u,u)-B(u,A^{\frac{s}{2}}e^{\alpha A^{\frac{1}{2}}}u),w\right),$$
	for $\|w\|_{L^2}=1.$ (As before, taking $w\in D(Gv(s, \alpha))$ with $\|w\|_{L^2}=1$, and then pass to the limit).
	
	Using the Fourier expansion of $u\&w$ are given by $$u=\sum_{j\in \mathbb{Z}^3\setminus \left \{(0,0,0)\right \}}  \hat{u}_j e^{ij \cdot x},\ w=\sum_{k\in \mathbb{Z}^3\setminus \left \{(0,0,0)\right \}}  \hat{w}_k e^{ik \cdot x}.$$ 
	
	It follows that
	\begin{align*}
	\left(A^{\frac{s}{2}}e^{\alpha A^{\frac{1}{2}}}B(u,u),w\right)&=\left(B(u,u),A^{\frac{s}{2}}e^{\alpha A^{\frac{1}{2}}}w\right)\\
	&=i \sum_{j,k} \left(j  \cdot \hat{u}_{k-j}\right)(\hat{u}_{j} \cdot \hat{w}_{-k})|k |^s e^{\alpha|k |},
	\end{align*}
	and
	\begin{align*}
	\left(B(u,A^{\frac{s}{2}}e^{\alpha A^{\frac{1}{2}}}u),w\right)=i \sum_{j,k} (j  \cdot \hat{u}_{k-j})(\hat{u}_{j} \cdot \hat{w}_{-k})|j |^s e^{\alpha|j |}.
	\end{align*}
	
	Combining the above two equations together, we have 
	$\displaystyle
	I=i \sum_{j,k} (j  \cdot \hat{u}_{k-j})(\hat{u}_{j} \cdot \hat{w}_{-k})\left(|k |^s e^{\alpha|k |}-|j |^s e^{\alpha|j |}\right).
	$
	Using the reality condition $\hat{w}_{-k}=\overline{\hat{w}_{k}}$, we obtain an estimate for $I$ given by
	\begin{align}
	\label{Iineq-n}
	|I| \leq  \sum_{j,k} |j | |\hat{u}_{k-j}| |\hat{u}_{j}| |\hat{w}_{k}| \left | |k |^s e^{\alpha|k |}-|j |^s e^{\alpha|j |}\right |.
	\end{align}
	
	Define $f$ by $f(x)=x^s e^{\alpha x}$. Then $f'(x)=s x^{s-1} e^{\alpha x}+x^{s} \alpha e^{\alpha x}$. Taking $\eta=a|j |+(1-a)|k |$, where $0\leq a \leq 1$, then $\eta$ is between $|j |$ and $|k |$. If $|k |\leq|j |$, then $|\eta|\leq|j |\leq|j |+|(k-j) |$; if $|j |<|k |$, then $|\eta|\leq|k |\leq|j |+|(k-j) |$. Therefore, we have $0<\eta\leq |j |+|(k-j) |$. Also, when $s\geq 1$, $s-1\geq 0$. Therefore, after applying the mean value theorem and the triangle inequality, it follows that
	\begin{align*}
	\left | |k |^s e^{\alpha|k |}-|j |^s e^{\alpha|j |}\right |&=|f'(\eta)| \left | |k |-|j | \right |\\
	&\leq  |f'(\eta)| |(k-j) |\\
	&=\left | s {\eta}^{s-1}e^{\alpha \eta}+{\eta}^{s}\alpha e^{\alpha \eta} \right | \left |(k-j) \right|\\
	&=\left | {\eta}^{s-1}e^{\alpha \eta}(s+\alpha\eta) \right | \left|(k-j) \right|.
	\end{align*}
	
	Replacing $\eta$ by $|j |+|l |$ with $l=k-j$, we have
	\begin{align}
	\label{est1a}
	&	\left | |k |^s e^{\alpha|k |}-|j |^s e^{\alpha|j |}\right |\\ \nonumber
	&\leq \left(|j |+|l |\right)^{s-1} e^{\alpha|j |}e^{\alpha|l |} \left(s+\alpha|j |+\alpha|l |\right) |l |.
	\end{align}
	
	Substituting (\ref{est1a}) into (\ref{Iineq-n}), we can refine our estimate for $I$
	\allowdisplaybreaks
	\begin{align*}
	|I| &\leq\sum_{l,j} |j |	|\hat{u}_l\|\hat{u}_j\|\hat{w}_{l+j}| (|j |+|l |)^{s-1}e^{\alpha(|j |+|l |)} \left(s+\alpha(|j |+|l |)\right) |l |\\
	&=	s\sum_{l,j} 	|\hat{u}_l\|\hat{u}_j\|\hat{w}_{l+j}| |l ||j |(|j |+|l |)^{s-1}e^{\alpha|j |}e^{\alpha|l |} \\
	&\quad +\alpha\sum_{l,j} 	|\hat{u}_l\|\hat{u}_j\|\hat{w}_{l+j}| |l ||j |(|j |+|l |)^{s}e^{\alpha|j |}e^{\alpha|l |} \\
	&\leq c_s\sum_{l,j}|\hat{u}_l\|\hat{u}_j\|\hat{w}_{l+j}| |l ||j |(|j |^{s-1}+|l |^{s-1})e^{\alpha|j |}e^{\alpha|l |} \\
	&\quad +c_s\alpha\sum_{l,j}|\hat{u}_l\|\hat{u}_j\|\hat{w}_{l+j}| |l ||j |(|j |^{s}+|l |^{s})e^{\alpha|j |}e^{\alpha|l |} \\
	&\leq  c_s\sum_{l,j}|\hat{u}_l\|\hat{u}_j\|\hat{w}_{l+j}| |l |^{s}|j | e^{\alpha|j |}e^{\alpha|l |} \\
	&
	+c_s\alpha\sum_{l,j}|\hat{u}_l\|\hat{u}_j\|\hat{w}_{l+j}\|l |^{s+1}|j |e^{\alpha|j |}e^{\alpha|l |} \\
	&\leq  c_s\sum_{j} |j |e^{\alpha|j |}|\hat{u}_j| \sum_{l}|l |^{s}  e^{\alpha|l |}|\hat{u}_l\|\hat{w}_{l+j}|  \\
	&
	+c_s\alpha \sum_{j} |j |e^{\alpha|j |}|\hat{u}_j| \sum_{l}|l |^{s+1}e^{\alpha|l |}|\hat{u}_l\|\hat{w}_{l+j}|\\
	&\leq  c_{s}\|u\|_{s, \alpha} \|w\|_{L^2} \sum_{j} |j |e^{\alpha|j |}|\hat{u}_j| 
	+c_{s}\alpha \|u\|_{s+1,\alpha} \|w\|_{L^2} \sum_{j} |j |e^{\alpha|j |}|\hat{u}_j|\\
	&\leq  	c_{s}	\|u\|_{s, \alpha}  \|w\|_{L^2}	\|e^{\alpha A^{\frac{1}{2}}}u\|_{F^1} 
	+	c_{s}\alpha \|u\|_{s+1, \alpha}  \|w\|_{L^2} 	\|e^{\alpha A^{\frac{1}{2}}}u\|_{F^1}.
	\end{align*}
	
	%
    Therefore	
	\begin{align*}
	&\left|\left(B(u,u),A^{s}e^{2\alpha A^{\frac{1}{2}}}u\right)\right |\\
	&=\|A^{\frac{s}{2}}e^{\alpha A^{\frac{1}{2}}}B(u,u)-B(u,A^{\frac{s}{2}}e^{\alpha A^{\frac{1}{2}}}u)\|_{L^2} \cdot\|A^{\frac{s}{2}}e^{\alpha A^{\frac{1}{2}}}\|_{L^2}\\
	&\leq c_{s}\|e^{\alpha A^{\frac{1}{2}}}u\|_{F^1} \|u\|^2_{s, \alpha} +c_s \alpha \|e^{\alpha A^{\frac{1}{2}}}u\|_{F^1}\|u\|_{s+1, \alpha} \|u\|_{s, \alpha}.
	\end{align*}
	
	This establishes (\ref{term32}). Moreover, after applying Young's inequality, we obtain 
	\begin{align*}
	c_s \alpha \|e^{\alpha A^{\frac{1}{2}}}u\|_{F^1}\|u\|_{s+1, \alpha} \|u\|_{s, \alpha}\leq {c_s \alpha^2} \|e^{\alpha A^{\frac{1}{2}}}u\|^2_{F^1}\|u\|^2_{s, \alpha}+\frac{1}{2}\|u\|^2_{s+1, \alpha}.
	\end{align*}
	
	Therefore,
	\begin{align*}
	\left|\left(B(u,u),A^{s}e^{2\alpha A^{\frac{1}{2}}}u\right)\right |\leq 
	c_{s}\|e^{\alpha A^{\frac{1}{2}}}u\|_{F^1} \|u\|^2_{s, \alpha}
	+{c_s \alpha^2} \|e^{\alpha A^{\frac{1}{2}}}u\|^2_{F^1}\|u\|^2_{s, \alpha}+\frac{1}{2}\|u\|^2_{s+1, \alpha},
	\end{align*}
	which is precisely (\ref{term33}).
\end{proof}
\textbf{Proof of Lemma \ref{lem33}}.
\begin{proof}
	For $\forall m>0$, if $0\leq  \alpha |k | \leq 1$, then $e^{ \alpha |k |}\leq e$, and if $ \alpha |k | \geq 1$, we have $e^{ \alpha |k |}\leq ( \alpha |k |)^m e^{ \alpha |k |}$. Therefore, for $\forall t>0$ and $k$, we have
	$\displaystyle e^{ \alpha |k |}\leq e+( \alpha |k |)^m e^{ \alpha |k |}$
	and
	$\displaystyle e^{2 \alpha |k |}\leq e+(2 \alpha |k |)^m e^{2 \alpha |k |}.$
	
	Taking $m=2s$, it follows that
	\begin{align*}
	\|e^{ \alpha A^{\frac{1}{2}}}u\|^{2}_{L^2}=\sum_k e^{2 \alpha |k |} |\hat{u}_k|^2 &\leq \sum_k \left ( e+(2 \alpha |k |)^{2s} e^{2 \alpha |k |}\right ) |\hat{u}_k|^2\\
	&=\sum_k e |\hat{u}_k|^2+\sum_k (2 \alpha |k |)^{2s} e^{2 \alpha |k |} |\hat{u}_k|^2.
	\end{align*}
	
	Since $\sqrt{a+b}\leq \sqrt{a}+\sqrt{b}$, for $a,b\geq0$, we have
	\begin{align*}
	\|e^{ \alpha A^{\frac{1}{2}}}u\|_{L^2}&\leq \sqrt{\sum_k e |\hat{u}_k|^2+\sum_k (2 \alpha |k |)^{2s} e^{2 \alpha |k |} |\hat{u}_k|^2}\\
	&\leq \sqrt{\sum_k e |\hat{u}_k|^2}+\sqrt{\sum_k (2 \alpha |k |)^{2s} e^{2 \alpha |k |} |\hat{u}_k|^2}\\
	&=\sqrt{e}\|u\|_{L^2}+\sqrt{ (2 \alpha )^{2s}\sum_k  |k |^{2s} e^{2 \alpha |k |} |\hat{u}_k|^2}\\
	&=\sqrt{e}\|u\|_{L^2}+(2 \alpha )^s ||A^{\frac{s}{2}}e^{ \alpha A^{\frac{1}{2}}}u||_{L^2}\\
	&=\sqrt{e}\|u\|_{L^2}+(2 \alpha )^s \|u\|_{s,  \alpha }.
	\end{align*}
\end{proof}

{\bf Proof of Lemma \ref{lemmaonzeta}.}
\begin{proof}
		Comparing the terms on the right hand side of (\ref{mainequ1}), we can expect that there is a region (when $t$ and $\zeta$ are both small) where $c_{s} \gamma \zeta^{1+\frac{5}{2s}}$ is the dominating term among the four terms on the right hand side. 
	In order to find this specific region, we compare $c_{s} \gamma \zeta^{1+\frac{5}{2s}}$ with the other three terms (note that $c_{s}$ is positive).
	\begin{enumerate}
		\label{comparison}
		\item Comparing $c_{s} \gamma \zeta^{1+\frac{5}{2s}} $ with $c_{s}  (\beta t)^{s-\frac{5}{2}}\zeta^{2}$:\\
		if $\displaystyle c_{s} \gamma \zeta^{1+\frac{5}{2s}} \geq c_{s}  (\beta t)^{s-\frac{5}{2}}\zeta^{2}$, equivalently, $\zeta \leq \frac{c_s \gamma^{\frac{2s}{2s-5}}}{(\beta t)^{s}}$.		\item Comparing $c_{s} \gamma \zeta^{1+\frac{5}{2s}}$ with $\displaystyle {c_s (\beta t)^2} \gamma^{2}\zeta^{1+\frac{5}{s}} $:\\ if $\displaystyle c_{s} \gamma \zeta^{1+\frac{5}{2s}}\geq {c_s (\beta t)^2} \gamma^{2}\zeta^{1+\frac{5}{s}}$, equivalently, $\displaystyle \zeta \leq \frac{c_s }{\gamma^{\frac{2s}{5}}(\beta t)^{\frac{4 s}{5}}}$.
		\item Comparing $c_{s} \gamma \zeta^{1+\frac{5}{2s}}$ with $\displaystyle {c_s (\beta t)^{2s-3}}  \zeta^{3}$:\\ if $\displaystyle c_{s} \gamma \zeta^{1+\frac{5}{2s}}\geq {c_s (\beta t)^{2s-3}}  \zeta^{3}$, equivalently, $ \displaystyle \zeta \leq \frac{c_s \gamma^{\frac{2s}{4s-5}}}{(\beta t)^{\frac{2s(2s-3)}{4s-5}}}$.
	\end{enumerate}
%
%
%
%
	
	Therefore, if $\displaystyle {\zeta} \leq c_s  \min \left\{\beta^{-s}, \beta^{-\frac{4s}{5}}, \beta^{-\frac{2s(2s-3)}{4s-5}}\right \} \cdot \min \left\{\gamma^{\frac{2s}{2s-5}}, \gamma^{-\frac{2s}{5}}, \gamma^{\frac{2s}{4s-5}}\right \} \cdot \min \left\{\frac{1}{{t}^{s}}, \frac{1}{{t}^{\frac{4s}{5}}}, \frac{1}{{t}^{\frac{2s(2s-3)}{4s-5}}}\right \}$, then the first term ($\displaystyle c_{s} \gamma \zeta^{1+\frac{5}{2s}}$)
	is the dominating term among the four terms on the right hand side of (\ref{mainequ1}).\\
	
	When $s>\frac{5}{2}$, we have $\frac{4s}{5}<\frac{2s \cdot (2s-3)}{4s-5}<s$. Therefore, when $\beta\leq \frac{1}{2}$,
	$\displaystyle \beta^{-\frac{4s}{5}}=\min \left\{\beta^{-s}, \beta^{-\frac{4s}{5}}, \beta^{-\frac{2s(2s-3)}{4s-5}}\right \}.$
	
	Denoting
	$$\tilde{c}=c_s \beta^{-\frac{4s}{5}} \min \left\{\gamma^{\frac{2s}{2s-5}}, \gamma^{-\frac{2s}{5}}, \gamma^{\frac{2s}{4s-5}}\right \}=c_s \beta^{-\frac{4s}{5}} \min \left\{\gamma^{\frac{2s}{2s-5}}, \gamma^{-\frac{2s}{5}}\right \}.$$
	
	When $0<{t}<1$:
	$\displaystyle
	\frac{1}{{t}^{\frac{4s}{5}}}=\min \left\{\frac{1}{{t}^{s}}, \frac{1}{{t}^{\frac{4s}{5}}}, \frac{1}{{t}^{\frac{2s(2s-3)}{4s-5}}}\right \}.
$ When ${t}>1$:
	$\displaystyle
	\frac{1}{{t}^{s}}=\min \left\{\frac{1}{{t}^{s}}, \frac{1}{{t}^{\frac{4s}{5}}}, \frac{1}{{t}^{\frac{2s(2s-3)}{4s-5}}}\right \}.
$
	
	From (\ref{mainequ1}), we observe that $\zeta$ starts with positve initial data and is an increasing function. Moreover, since $\zeta \nearrow \infty$ as $t\nearrow T_{\zeta}$, it will first intersect either the curve $\displaystyle \frac{\tilde{c}}{{t}^{\frac{4s}{5}}}$ or the curve $\displaystyle \frac{\tilde{c}}{{t}^{s}}$ for some $t_{\zeta}\in (0, T_{\zeta})$. We have the following cases.
	
	Case (i): when $\displaystyle {\zeta}({0})\geq \tilde{c} $, then $\displaystyle {\zeta}({1})> \tilde{c}$. In this case, $ \zeta(t)$ first intercepts with the curve of $\displaystyle \frac{\tilde{c}}{{t}^{\frac{4s}{5}}}$. Denoting the interception point as ${t}_{\zeta}$, we have 
	$0<{t}_{\zeta}\leq1$.	
	
	
	Therefore, when $0<t<t_{\zeta}$, we have $\displaystyle {\zeta}(t) <\frac{\tilde{c}}{{t}^{\frac{4s}{5}}}$. ($\displaystyle c_{s} \gamma \zeta^{1+\frac{5}{2s}}$)
	is the dominating term among the four terms on the right hand side of (\ref{mainequ1}). It follows that
	\begin{align}
	\label{zetaa1}
	\frac{d {\zeta}}{dt}<4c_s \gamma {\zeta}^{1+\frac{5}{2s}}.
	\end{align}
	
	Moreover, when $0<t<t_{\zeta}$, we compare $\zeta(t)$ with $\phi(t)$, where, $\phi(t)$ is the solution of
	\begin{align}
	\label{phi1}
	\frac{d \phi}{dt}=4c_s \gamma {\phi}^{1+\frac{5}{2s}},
	\end{align}
	with $\phi(0)=\zeta(0)$.
	
	Applying Lemma \ref{lemma:nonlinear_Gronwall} on (\ref{zetaa1}) and (\ref{phi1}), we have:
$\displaystyle
	\zeta(t) < \phi(t),\ \text{for\ all}\ t\in \left[0, \min \left\{t_{\zeta}, T_{\zeta}, T_{\phi}\right\}\right].
$
	
	It follows that
	there exists a $t_{\phi}$ that 
	\begin{align}
	\label{tphi}
	\phi(t_{\phi})=\frac{\tilde{c}}{t_{\phi}^\frac{4s}{5}}.
	\end{align} 
	
	Since $\zeta(t)< \phi(t)$, we conclude $0<t_{\phi}< t_{\zeta}\leq 1$. Thus, the following relation holds: $\displaystyle t_{\phi}< t_{\zeta}<T_{\zeta}.$
	%
	
	Solving (\ref{phi1}), we have
	\begin{align}
	\label{tsol}
	\phi(t)=(\phi(0)^{-\frac{5}{2s}}-c_s \gamma  t)^{-\frac{2s}{5}}.
	\end{align}
	
	Combining (\ref{tphi}) and (\ref{tsol}), it holds that:
$\displaystyle
	(\phi(0)^{-\frac{5}{2s}}-c_s \gamma  t_{\phi})^{-\frac{2s}{5}}=\tilde{c} t_{\phi}^{-\frac{4s}{5}}.
	$
	
	After simplification, we obtain:
	$\displaystyle
	\phi(0)^{-\frac{5}{2s}}-c_s \gamma  t_{\phi}=\tilde{c}^{\frac{-5}{2s}}t_{\phi}^{2}.
$
	
	Therefore
	\begin{align}
	\label{equt}
	\tilde{c}^{\frac{-5}{2s}}t_{\phi}^2+c_s \gamma  t_{\phi}=\phi(0)^{-\frac{5}{2s}}.
	\end{align}
	
	Since $t_{\phi}<1$, i.e., ${t_{\phi}}^2< t_{\phi}$, from (\ref{equt}), we have:
$\displaystyle
	\phi(0)^{-\frac{5}{2s}}< \left(\tilde{c}^{\frac{-5}{2s}}+c_s \gamma \right)  t_{\phi}.
$
	
	Therefore
	\begin{align}
	\label{equtphi}
	\phi(0)^{-\frac{5}{2s}}<\frac{1}{\tilde{\tilde{c}}} t_{\phi},
	\end{align}
	where $\displaystyle \frac{1}{\tilde{\tilde{c}}}=2\max \left\{\tilde{c}^{\frac{-5}{2s}}, c_s \gamma \right \}$. Since $\tilde{c}=c_s \beta^{-\frac{4s}{5}} \min \left\{\gamma^{\frac{2s}{2s-5}}, \gamma^{-\frac{2s}{5}}\right \}$, therefore $$\tilde{c}^{\frac{-5}{2s}}=c_s \beta^2 \max \left\{\gamma^{-\frac{5}{2s-5}}, \ \gamma\right \}.$$ 
	
%
%
%
Since $\beta<1$, we have $\displaystyle \frac{1}{\tilde{\tilde{c}}}=c_s \max \left\{\gamma^{-\frac{5}{2s-5}}, \ \gamma\right \}$, i.e., $\displaystyle\displaystyle {\tilde{\tilde{c}}}=c_s \min \left\{\gamma^{\frac{5}{2s-5}}, \ \gamma^{-1}\right \}.$
	
	From (\ref{equtphi}), we have
	\begin{align*}
	t_{\phi}>\frac{\tilde{\tilde{c}}}{\phi(0)^{\frac{5}{2s}}}=\frac{c_s \min \left\{\gamma^{\frac{5}{2s-5}}, \ \gamma^{-1}\right \}}{\zeta(0)^{\frac{5}{2s}}}.
	\end{align*}
	
	Therefore
	\begin{align*}
T_{\zeta}>	t_{\phi}>\frac{c_s \min \left\{\gamma^{\frac{5}{2s-5}}, \ \gamma^{-1}\right \}}{\zeta(0)^{\frac{5}{2s}}}.
	\end{align*}
		Case (ii):  when $\displaystyle {\zeta}({0})< \tilde{c} $, if $\displaystyle {\zeta}({1})\geq \tilde{c}$, same as Case (i), we have
			$\displaystyle
		T_{\zeta}>	t_{\phi}>\frac{c_s \min \left\{\gamma^{\frac{5}{2s-5}}, \ \gamma^{-1}\right \}}{\zeta(0)^{\frac{5}{2s}}}.
	$
		
		If $\displaystyle {\zeta}({1})< \tilde{c}$, in this case, $ \zeta(t)$ first intercepts with the curve of $\displaystyle \frac{\tilde{c}}{{t}^{s}}$. Denoting the interception point as ${t}_{\zeta}$, we have 
		${t}_{\zeta}>1$.	
	
	Similar to Case (i), we have: when $0<t<t_{\zeta}$, 
	$\displaystyle
	\frac{d {\zeta}}{dt}<4c_s \gamma {\zeta}^{1+\frac{5}{2s}},
	$ Also, when we consider $\phi(t)$ as the solution of
	\begin{align}
	\label{newphi}
	\frac{d \phi}{dt}=4c_s \gamma {\phi}^{1+\frac{5}{2s}},
	\end{align}
	with $\phi(0)=\zeta(0)$, we have:
$\displaystyle
	\zeta(t) < \phi(t), \text{for\ all}\ t\in \left[0, \min \left\{t_{\zeta}, T_{\zeta}, T_{\phi}\right\}\right].
	$
	
	Moreover,
	$\displaystyle t_{\phi}< t_{\zeta}<T_{\zeta}.$	If $0<t_{\phi}\leq1$, same as Case (i), we have
	\begin{align*}
	T_{\zeta}>	t_{\phi}>\frac{c_s \min \left\{\gamma^{\frac{5}{2s-5}}, \ \gamma^{-1}\right \}}{\zeta(0)^{\frac{5}{2s}}}.
	\end{align*}
	
	If $t_{\phi}>1$, then
	\begin{align}
	\label{tphi2}
	\phi(t_{\phi})=\frac{\tilde{c}}{t_{\phi}^{s}}.
	\end{align} 
	
	Solving (\ref{newphi}), we have
	\begin{align}
	\label{tsol1}
	\phi(t)=(\phi(0)^{-\frac{5}{2s}}-c_s \gamma  t)^{-\frac{2s}{5}}.
	\end{align}
	
	Combining (\ref{tphi2}) and (\ref{tsol1}), we have:
$\displaystyle
	(\phi(0)^{-\frac{5}{2s}}-c_s \gamma  t_{\phi})^{-\frac{2s}{5}}=\tilde{c} t_{\phi}^{-s}.
	$
	
	After simplification, we obtain:
	$\displaystyle
	\phi(0)^{-\frac{5}{2s}}-c_s \gamma  t_{\phi}=\tilde{c}^{\frac{-5}{2s}}t_{\phi}^{5/2}.
$ Therefore
	\begin{align}
	\label{equt1}
	\tilde{c}^{\frac{-5}{2s}}t_{\phi}^{5/2}+c_s \gamma  t_{\phi}=\phi(0)^{-\frac{5}{2s}}.
	\end{align}
	
	Since $t_{\phi}>1$, then ${t_{\phi}}^{5/2}> t_{\phi}$, from (\ref{equt1}), we have:
$\displaystyle
	\phi(0)^{-\frac{5}{2s}}< \left(\tilde{c}^{\frac{-5}{2s}}+c_s \gamma \right)  {t_{\phi}}^{5/2}.
$ Therefore
	\begin{align}
	\label{equtphi1}
	\phi(0)^{-\frac{5}{2s}}<\frac{1}{\tilde{\tilde{c}}} {t_{\phi}}^{5/2}.
	\end{align}
	
	Following similar analysis as Case (i), we have $\displaystyle {\tilde{\tilde{c}}}=c_s \min \left\{\gamma^{\frac{5}{2s-5}}, \ \gamma^{-1}\right \}$ and
	\begin{align*}
 T_{\zeta}>t_{\phi}>\frac{\tilde{\tilde{c}}^{2/5}}{\phi(0)^{\frac{1}{s}}}=\frac{c_s \min \left\{\gamma^{\frac{2}{2s-5}}, \ \gamma^{-2/5}\right \}}{\zeta(0)^{\frac{1}{s}}}.
	\end{align*}
	
	Therefore, for Case (ii), we have
		\begin{align*}
	T_{\zeta}>t_{\phi}>\min\left\{Z, Z^{2/5}\right\},
	\end{align*}
	where $\displaystyle Z=\frac{c_s \min \left\{\gamma^{\frac{5}{2s-5}}, \ \gamma^{-1}\right \}}{\zeta(0)^{\frac{5}{2s}}}.$
	
	\end{proof}
\textbf{Proof of Lemma \ref{inequ3w}}.
\begin{proof}
	\begin{align}
		\label{bwu}
		\left |	\left (B({\omega}, u), A^{\tilde{s}}e^{2\alpha A^{\frac{1}{2}}}{\omega} \right ) \right |&=\left |\left(A^{\frac{\tilde{s}}{2}} e^{\alpha A^{\frac{1}{2}}}B({\omega}, u), A^{\frac{\tilde{s}}{2}} e^{\alpha A^{\frac{1}{2}}} {\omega}  \right )\right |\\ \nonumber
		&\leq \|{\omega} \cdot \nabla u\|_{\tilde{s}, \alpha } \|{\omega}\|_{\tilde{s}, \alpha }.
	\end{align}
	
	When $-\frac{1}{2}<\tilde{s}<\frac{3}{2}$, applying Lemma \ref{lem-gev1} with $s_1=\frac{3+2\tilde{s}}{4}$ and $s_2=\frac{3+2\tilde{s}}{4}$, we have:
$\displaystyle
		\|{\omega} \cdot \nabla u\|_{\tilde{s}, \alpha }\leq c_{\tilde{s}}\|{\omega}\|^2_{\frac{3+2\tilde{s}}{4}, \alpha }.
$
	
	Furthermore, $\displaystyle
		\|{\omega}\|^2_{\frac{3+2\tilde{s}}{4}, \alpha } \leq c_{\tilde{s}}\|{\omega}\|^{\frac{1+2\tilde{s}}{2}}_{\tilde{s}, \alpha } \|{\omega}\|^{\frac{3-2\tilde{s}}{2}}_{\tilde{s}+1, \alpha }.
$ Therefore, (\ref{bwu}) beomes $$\left|	\left (B({\omega}, u), A^{\tilde{s}}e^{2\alpha A^{\frac{1}{2}}}{\omega} \right ) \right |\leq c_{\tilde{s}} \|{\omega}\|_{\tilde{s}, \alpha }^{\frac{3+2\tilde{s}}{2}} \|{\omega}\|_{\tilde{s}+1, \alpha }^{\frac{3-2\tilde{s}}{2}} .$$
\end{proof}
\textbf{Proof of Lemma \ref{inequ4w}}.
\begin{proof}
	Starting from
	\begin{align*}
		\left(B(u, {\omega}),A^{\tilde{s}}e^{2\alpha A^{\frac{1}{2}}}{\omega}\right)=\left(A^{\frac{\tilde{s}}{2}}e^{\alpha A^{\frac{1}{2}}}B(u,{\omega}),A^{\frac{\tilde{s}}{2}}e^{\alpha A^{\frac{1}{2}}}{\omega}\right).
	\end{align*}
	
	Since $\left(B(u,A^{\frac{\tilde{s}}{2}}e^{\alpha A^{\frac{1}{2}}}{\omega}),A^{\frac{\tilde{s}}{2}}e^{\alpha A^{\frac{1}{2}}}{\omega}\right)=0$, it follows that
	\begin{align}
		\label{estB-n}
		\left(B(u,{\omega}),A^{\tilde{s}}e^{2\alpha A^{\frac{1}{2}}}{\omega}\right)=\left(A^{\frac{\tilde{s}}{2}}e^{\alpha A^{\frac{1}{2}}}B(u,{\omega})-B(u,A^{\frac{\tilde{s}}{2}}e^{\alpha A^{\frac{1}{2}}}{\omega}),A^{\frac{\tilde{s}}{2}}e^{\alpha A^{\frac{1}{2}}}{\omega}\right)= P.
	\end{align}
	
	Futhermore
	\begin{align*}
		\left(A^{\frac{\tilde{s}}{2}}e^{\alpha A^{\frac{1}{2}}}B(u,{\omega}),A^{\frac{\tilde{s}}{2}}e^{\alpha A^{\frac{1}{2}}}{\omega}\right)&=\left(B(u,{\omega}),A^{\tilde{s}}e^{2\alpha A^{\frac{1}{2}}}\omega\right)\\
		&=i \sum_{j,k} (j  \cdot \hat{u}_{k-j})(\hat{\omega}_{j} \cdot \hat{\omega}_{-k})|k |^{2\tilde{s}} e^{2\alpha |k |},
	\end{align*}
	and
	\begin{align*}
		\left(B(u,A^{\frac{\tilde{s}}{2}}e^{\alpha A^{\frac{1}{2}}}{\omega}),A^{\frac{\tilde{s}}{2}}e^{\alpha A^{\frac{1}{2}}}{\omega}\right)=i \sum_{j,k} (j  \cdot \hat{u}_{k-j})(\hat{\omega}_{j} \cdot \hat{\omega}_{-k})|j |^{\tilde{s}} e^{\alpha |j |}|k |^{\tilde{s}} e^{\alpha |k |}.
	\end{align*}
	
	Combining the above two equations together, we have 
	\begin{align*}
		P=i \sum_{j,k} (j  \cdot \hat{u}_{k-j})(\hat{\omega}_{j} \cdot \hat{\omega}_{-k})|k |^{\tilde{s}} e^{\alpha |k |}\left(|k |^{\tilde{s}} e^{\alpha |k |}-|j |^{\tilde{s}} e^{\alpha |j |}\right).
	\end{align*}
	
	Since $u$ is divergence free, we have $(k-j)  \cdot \hat{u}_{k-j}=0$ and
	\begin{align*}
		P=i \sum_{j,k} (k  \cdot \hat{u}_{k-j})(\hat{\omega}_{j} \cdot \hat{\omega}_{-k})|k |^{\tilde{s}} e^{\alpha |k |}\left(|k |^{\tilde{s}} e^{\alpha |k |}-|j |^{\tilde{s}} e^{\alpha |j |}\right).
	\end{align*}
	
	Since $\hat{\omega}_{-k}=\overline{\hat{\omega}_{k}}$, we obtain the estimate of $P$
	\begin{align}
		\label{Iineq}
		|P| \leq  \sum_{j,k} |k | |\hat{u}_{k-j}| |\hat{\omega}_{j}| |\hat{\omega}_{k}||k |^{\tilde{s}} e^{\alpha |k |} \left | |k |^{\tilde{s}} e^{\alpha |k |}-|j |^{\tilde{s}} e^{\alpha |j |}\right |.
	\end{align}
	
	Defining $f(x)=x^{\tilde{s}} e^{\alpha x}$, then $f'(x)=\tilde{s} x^{{\tilde{s}}-1} e^{\alpha x}+x^{{\tilde{s}}} \alpha  e^{\alpha x}$. Taking $\eta=a|j |+(1-a)|k |$, where $0\leq a \leq 1$, then $\eta$ is between $|j |$ and $|k |$. If $|k |<|j |$, then $|\eta|<|j |<|j |+|(k-j) |$; if $|j |<|k |$, then $|\eta|<|k |\leq|j |+|(k-j) |$. Therefore, we have $0<\eta\leq |j |+|(k-j) |$. Applying the mean value theorem, it follows that
	\begin{align*}
		\left | |k |^{\tilde{s}} e^{\alpha |k |}-|j |^{\tilde{s}} e^{\alpha |j |}\right |=|f'(\eta)| \left | |k |-|j | \right | &\leq  |f'(\eta)| |(k-j) |\\
		&=\left | (\tilde{s} {\eta}^{{\tilde{s}}-1}e^{\alpha  \eta}+{\eta}^{{\tilde{s}}}\alpha  e^{\alpha  \eta}) \right | \left |(k-j) \right |.
	\end{align*}
	
	Therefore, taking $l=k-j$, (\ref{Iineq}) becomes
	\begin{align*}
		|P|&\leq \sum_{l+j=k} |k | |\hat{u}_{l}| |\hat{\omega}_{j}| |\hat{\omega}_{k}||k |^{\tilde{s}} e^{\alpha {|k |}} \left | (\tilde{s} {\eta}^{{\tilde{s}}-1}e^{\alpha  \eta}+{\eta}^{{\tilde{s}}}\alpha  e^{\alpha  \eta}) \right | |l |\\
		&\leq |\tilde{s}| \sum_{l+j=k} |k | |\hat{u}_{l}| |\hat{\omega}_{j}| |\hat{\omega}_{k}||k |^{\tilde{s}} e^{\alpha {|k |}} \left | {\eta}\right |^{{\tilde{s}}-1}e^{\alpha  \eta}  |l |\\
		&+\alpha  \sum_{l+j=k} |k | |\hat{u}_{l}| |\hat{\omega}_{j}| |\hat{\omega}_{k}||k |^{\tilde{s}} e^{\alpha {|k |}} \left | {\eta} \right |^{{\tilde{s}}}e^{\alpha  \eta} |l |\\
		&=P_1+P_2.
	\end{align*}
	
	We first analyze $\displaystyle P_1=|\tilde{s}| \sum_{l+j=k} |k | |\hat{u}_{l}| |\hat{\omega}_{j}| |\hat{\omega}_{k}||k |^{\tilde{s}} e^{\alpha {|k |}} \left | {\eta}\right |^{{\tilde{s}}-1}e^{\alpha  \eta}  |l |$.\\
	Case (i): When $-\frac{1}{2}<\tilde{s}<1$, since $\left | {\eta}\right |=a|j |+(1-a)|k |$, $0\leq a \leq 1$ and we have\\
	Case (ia): if $|j |\leq |k |$, then $\left | {\eta}\right |\geq |j |$, we have:
	$\displaystyle
		\left | {\eta}\right |^{{\tilde{s}}-1}\leq  |j |^{{\tilde{s}}-1}.
	$
	
	Moreover, since $0<\eta\leq |j |+|l |$, we have:
$\displaystyle
		e^{\alpha  \eta}\leq e^{\alpha |j |}e^{\alpha |l |}. 
$ Taking $0<\delta<1$, it follows that
	\begin{align*}
		P_1&\leq |\tilde{s}| \sum_{l+j=k}|k | |\hat{u}_{l}| |\hat{\omega}_{j}| |\hat{\omega}_{k}||k |^{\tilde{s}} e^{\alpha |k |} |j |^{{\tilde{s}}-1} e^{\alpha |j |}e^{\alpha |l |} |l | \\
		&\leq |\tilde{s}| \sum_{l+j=k}  |k |^{1-\delta} (|l ||\hat{u}_{l}|e^{\alpha |l |}) \cdot (|j |^{{\tilde{s}}-1}|\hat{\omega}_{j}|e^{\alpha |j |})\cdot (|\hat{\omega}_{k}||k |^{{\tilde{s}}+\delta} e^{\alpha |k |} ) \\
		&\leq |\tilde{s}| 		\|{\omega}_1\ast {\omega}_2\|_{\dot{H}^{1-\delta}} \|{\omega}\|_{\tilde{s}+\delta,  \alpha },
	\end{align*}
	where 
	\begin{align*}
		\|{\omega}_1\|^2_{L^2}=\sum_{l} |\hat{{\omega}}_{l}|^2e^{2\alpha |l |},\
		\|{\omega}_2\|^2_{L^2}=\sum_{l} |l |^{2(\tilde{s}-1)}|\hat{{\omega}}_{l}|^2e^{2\alpha |l |}.
	\end{align*}
	
	When $-\frac{1}{2}< \tilde{s}<1$ and $\max\left\{\frac{1}{2}-\tilde{s},0\right\}<\delta<1$, from Lemma \ref{lem-gev1} with $s_1=\frac{3-2\delta+2\tilde{s}}{4}$ and $s_2=\frac{7-2\delta-2\tilde{s}}{4}$, we have
	\begin{align*}
		\|{\omega}_1\ast {\omega}_2\|_{\dot{H}^{1-\delta}}\leq c_{\tilde{s}}\|{\omega}_1\|_{\frac{3-2\delta+2\tilde{s}}{4}} \|{\omega}_2\|_{\frac{7-2\delta-2\tilde{s}}{4}} =c\|{\omega}\|^2_{\frac{3+2\tilde{s}-2\delta}{4}, \alpha }.
	\end{align*}
	
	Therefore
	\begin{align*}
		P_1\leq c_{\tilde{s}} \|{\omega}\|^2_{\frac{3+2\tilde{s}-2\delta}{4}, \alpha }\|{\omega}\|_{\tilde{s}+\delta,  \alpha }.
	\end{align*}
	
	When $-\frac{1}{2}< \tilde{s}<1$ with $\max\left\{\frac{1}{2}-\tilde{s}, 0\right\}<\delta<\min\left\{\frac{3}{2}-\tilde{s}, 1\right\}$, we have
	\begin{align*}
		\|{\omega}\|^2_{\frac{3+2\tilde{s}-2\delta}{4}, \alpha }\leq c_{\tilde{s}} \|{\omega}\|^{\frac{2\delta+2\tilde{s}+1}{2}}_{\tilde{s}, \alpha } \|{\omega}\|^{\frac{3-2\delta-2\tilde{s}}{2}}_{\tilde{s}+1, \alpha }\ 
	\text{
	and}\ 
		\|{\omega}\|_{\tilde{s}+\delta, \alpha }\leq c_{\tilde{s}} \|{\omega}\|^{1-\delta}_{\tilde{s}, \alpha } \|{\omega}\|^{\delta}_{\tilde{s}+1, \alpha }.
	\end{align*}
	
	Therefore
	\begin{align*}
		P_1&\leq c_{\tilde{s}} \|{\omega}\|^{\frac{2\delta+2\tilde{s}+1}{2}}_{\tilde{s}, \alpha } \|{\omega}\|^{\frac{3-2\delta-2\tilde{s}}{2}}_{\tilde{s}+1, \alpha }\|{\omega}\|^{1-\delta}_{\tilde{s}, \alpha } \|{\omega}\|^{\delta}_{\tilde{s}+1, \alpha }\\ \nonumber
		&=c_{\tilde{s}} \|{\omega}\|^{\frac{3+2\tilde{s}}{2}}_{\tilde{s}, \alpha } \|{\omega}\|^{\frac{3-2\tilde{s}}{2}}_{\tilde{s}+1, \alpha }.
	\end{align*}
	Case (ib): if $|j |> |k |$, then $\left | {\eta}\right |\geq |k |$, we have:
$\displaystyle
		\left | {\eta}\right |^{{\tilde{s}}-1}\leq  |k |^{{\tilde{s}}-1}.
$
	
	Therefore
	\begin{align*}
		P_1&\leq |\tilde{s}|\sum_{l+j=k}|k | |\hat{u}_{l}| |\hat{\omega}_{j}| |\hat{\omega}_{k}||k |^{\tilde{s}} e^{\alpha |k |} |k |^{{\tilde{s}}-1} e^{\alpha |j |}e^{\alpha |l |} |l | \\
		&\leq |\tilde{s}|\sum_{l+j=k}  |k |^{\tilde{s}} (|l ||\hat{u}_{l}|e^{\alpha |l |}) \cdot (|\hat{\omega}_{j}|e^{\alpha |j |})\cdot (|\hat{\omega}_{k}||k |^{\tilde{s}} e^{\alpha |k |} ) \\
		&\leq |\tilde{s}| 		\|{\omega}_1\ast {\omega}_1\|_{\dot{H}^{\tilde{s}}} \|{\omega}\|_{\tilde{s},  \alpha }.
	\end{align*}
	
	When $-\frac{1}{2}<\tilde{s}<1$, from Lemma \ref{lem-gev1} with $s_1=\frac{3+2\tilde{s}}{4}$ and $s_2=\frac{3+2\tilde{s}}{4}$, we have:
$\displaystyle
		\|{\omega}_1\ast {\omega}_1\|_{\dot{H}^{\tilde{s}}}\leq c_{\tilde{s}}\|{\omega}\|^2_{\frac{3+2\tilde{s}}{4}, \alpha }.
$
	
	Therefore,
$\displaystyle
		P_1\leq c_{\tilde{s}} \|{\omega}\|^2_{\frac{3+2\tilde{s}}{4}, \alpha }\|{\omega}\|_{\tilde{s},  \alpha }.
	$
	
	Since
	\begin{align}
		\label{inter2}
		\|{\omega}\|^2_{\frac{3+2\tilde{s}}{4}, \alpha }\leq c_{\tilde{s}} \|{\omega}\|^{\frac{1+2\tilde{s}}{2}}_{\tilde{s}, \alpha } \|{\omega}\|^{\frac{3-2\tilde{s}}{2}}_{\tilde{s}+1, \alpha },
	\end{align}
	we have:
$\displaystyle
		P_1\leq c_{\tilde{s}} \|{\omega}\|^{\frac{3+2\tilde{s}}{2}}_{\tilde{s}, \alpha } \|{\omega}\|^{\frac{3-2\tilde{s}}{2}}_{\tilde{s}+1, \alpha }.
	$
	
	Case (ii): When $1\leq \tilde{s}<\frac{3}{2}$, since $|\eta|\leq |j |+|l |$, we have:
	$\displaystyle
		\left | {\eta}\right |^{{\tilde{s}}-1}\leq  (|j |+|l |)^{{\tilde{s}}-1}.
	$
	
	Therefore
	\begin{align*}
		P_1&\leq \tilde{s}\sum_{l+j=k}|k | |\hat{u}_{l}| |\hat{\omega}_{j}| |\hat{\omega}_{k}||k |^{{\tilde{s}}} e^{\alpha |k |} (|j |+|l |)^{{\tilde{s}}-1} e^{\alpha |j |}e^{\alpha |l |} |l | \\
		&\leq c_{\tilde{s}}\tilde{s}\sum_{l+j=k}  |k | (|l ||\hat{u}_{l}|e^{\alpha |l |}) \cdot (|j |^{{\tilde{s}}-1}+|l |^{{\tilde{s}}-1})\cdot (|\hat{\omega}_{j}|e^{\alpha |j |})\cdot (|\hat{\omega}_{k}||k |^{\tilde{s}} e^{\alpha |k |} ) \\
		&\leq c_{\tilde{s}}\tilde{s}\sum_{l+j=k}  |k | (|l ||\hat{u}_{l}|e^{\alpha |l |}) \cdot (|j |^{{\tilde{s}}-1}|\hat{\omega}_{j}|e^{\alpha |j |})\cdot (|\hat{\omega}_{k}||k |^{{\tilde{s}}} e^{\alpha |k |} )\\
		&\leq c_{\tilde{s}}	\tilde{s}	\|{\omega}_1\ast {\omega}_2\|_{\dot{H}^{1}} \|{\omega}\|_{\tilde{s},  \alpha }.
	\end{align*}
	
	When $1\leq \tilde{s}<\frac{3}{2}$, from Lemma \ref{lem-gev1} with $s_1=\frac{3+2\tilde{s}}{4}$ and $s_2=\frac{7-2\tilde{s}}{4}$, we have
	\begin{align*}
		\|{\omega}_1\ast {\omega}_2\|_{\dot{H}^1}\leq c_{\tilde{s}}\|{\omega}_1\|_{\frac{3+2\tilde{s}}{4}} \|{\omega}_2\|_{\frac{7-2\tilde{s}}{4}} =c_{\tilde{s}}\|{\omega}\|^2_{\frac{3+2\tilde{s}}{4}, \alpha }.
	\end{align*}
	
	Therefore,
	$\displaystyle
		P_1\leq c_{\tilde{s}} \|{\omega}\|^2_{\frac{3+2\tilde{s}}{4}, \alpha }\|{\omega}\|_{\tilde{s},  \alpha }.
	$
	
	From (\ref{inter2}), we have:
	$\displaystyle
		P_1\leq c_{\tilde{s}} \|{\omega}\|^{\frac{3+2\tilde{s}}{2}}_{\tilde{s}, \alpha } \|{\omega}\|^{\frac{3-2\tilde{s}}{2}}_{\tilde{s}+1, \alpha }.
	$
	
	Combining case (i) and (ii), when $-\frac{1}{2}<\tilde{s}<\frac{3}{2}$, we always have
	\begin{align}
		\label{p1case3}
		P_1\leq c_{\tilde{s}} \|{\omega}\|^{\frac{3+2\tilde{s}}{2}}_{\tilde{s}, \alpha } \|{\omega}\|^{\frac{3-2\tilde{s}}{2}}_{\tilde{s}+1, \alpha }.
	\end{align}
	
	Next, we can analyze the estimate for $$ P_2=\alpha  \sum_{l+j=k} |k | |\hat{u}_{l}| |\hat{\omega}_{j}| |\hat{\omega}_{k}||k |^{\tilde{s}} e^{\alpha {|k |}} \left | {\eta} \right |^{{\tilde{s}}}e^{\alpha  \eta} |l |.$$
	Case (a): $0\leq \tilde{s}<\frac{3}{2}$, we have:
	$\displaystyle
		\left | {\eta}\right |^{{\tilde{s}}}\leq  (|j |+|l |)^{{\tilde{s}}}.
$ Therefore,
	\begin{align*}
		P_2 & \leq \alpha \sum_{l+j=k} |k | |\hat{u}_{l}| |\hat{\omega}_{j}| |\hat{\omega}_{k}||k |^{\tilde{s}} e^{\alpha |k |} (|j |+|l |)^{{\tilde{s}}} e^{\alpha |j |}e^{\alpha |l |} |l |\ \\
		&\leq c_{\tilde{s}}\alpha \sum_{l+j=k}|k | |\hat{u}_{l}| |\hat{\omega}_{j}| |\hat{\omega}_{k}||k |^{\tilde{s}} e^{\alpha |k |} (|j |^{{\tilde{s}}}+|l |^{{\tilde{s}}}) e^{\alpha |j |}e^{\alpha |l |} |l |\ \\
		&\leq c_{\tilde{s}}\alpha \sum_{l+j=k}|k | |\hat{u}_{l}| |\hat{\omega}_{j}| |\hat{\omega}_{k}||k |^{\tilde{s}} e^{\alpha |k |} |l |^{{\tilde{s}}}e^{\alpha |j |}e^{\alpha |l |} |l |\\
		&\leq c_{\tilde{s}}\alpha  \sum_{l+j=k}|k |^{1-\delta}|l |^{{\tilde{s}}}  (|l ||\hat{u}_{l}|e^{\alpha |l |})\cdot (|\hat{\omega}_{j}|e^{\alpha |j |})\cdot (|\hat{\omega}_{k}||k |^{{\tilde{s}}+\delta} e^{\alpha |k |} ) \\
		&\leq  c_{\tilde{s}}\alpha  	\|{\omega}_1\ast {\omega}_3\|_{\dot{H}^{1-\delta}} \|{\omega}\|_{\tilde{s}+\delta,  \alpha },
	\end{align*}
	where 
$\displaystyle
		\|{\omega}_3\|^2_{L^2}=\sum_{l} |l |^{2\tilde{s}}|\hat{{\omega}}_{l}|^2e^{2\alpha |l |}.
	$ When $0\leq \tilde{s}< \frac{3}{2}$ with $\max\left\{\tilde{s}-\frac{1}{2},0\right\}<\delta<1$, from Lemma \ref{lem-gev1} with $\displaystyle s_1=\frac{5+2\tilde{s}-2\delta}{4},$ and $\displaystyle s_2=\frac{5-2\delta-2\tilde{s}}{4}.$ We have
	\begin{align*}
		\|{\omega}_1\ast {\omega}_3\|_{\dot{H}^{1-\delta}}\leq c_{\tilde{s}}\|{\omega}_1\|_\frac{5+2\tilde{s}-2\delta}{4} \|{\omega}_3\|_\frac{5-2\delta-2\tilde{s}}{4} =c_{\tilde{s}}\|{\omega}\|^2_{\frac{5+2\tilde{s}-2\delta}{4}, \alpha }.
	\end{align*}
	
	Therefore,
	$\displaystyle
		P_2\leq c_{\tilde{s}} \alpha  \|{\omega}\|^2_{\frac{5+2\tilde{s}-2\delta}{4}, \alpha } \|{\omega}\|_{\tilde{s}+\delta,\alpha }.
$
	
	When $0\leq \tilde{s}< \frac{3}{2}$ with $\max\left\{\tilde{s}-\frac{1}{2}, \frac{1}{2}-\tilde{s}, 0\right\}<\delta<1$, we have
	\begin{align*}
		\|{\omega}\|^2_{\frac{5+2\tilde{s}-2\delta}{4}, \alpha }\leq c_{\tilde{s}} \|{\omega}\|^{\frac{2\delta+2\tilde{s}-1}{2}}_{\tilde{s}, \alpha } \|{\omega}\|^{\frac{5-2\delta-2\tilde{s}}{2}}_{\tilde{s}+1, \alpha } \quad \mbox{and}\quad 
		\|{\omega}\|_{\tilde{s}+\delta, \alpha }\leq c_{\tilde{s}} \|{\omega}\|^{1-\delta}_{\tilde{s}, \alpha } \|{\omega}\|^{\delta}_{\tilde{s}+1, \alpha }.
	\end{align*}
	\comments{
	and
	\begin{align*}
		\|{\omega}\|_{\tilde{s}+\delta, \alpha }\leq c_{\tilde{s}} \|{\omega}\|^{1-\delta}_{\tilde{s}, \alpha } \|{\omega}\|^{\delta}_{\tilde{s}+1, \alpha }.
	\end{align*}
	}
	Thus
	\begin{align*}
		P_2&\leq c_{\tilde{s}} \alpha \|{\omega}\|^{\frac{2\delta+2\tilde{s}-1}{2}}_{\tilde{s}, \alpha } \|{\omega}\|^{\frac{5-2\delta-2\tilde{s}}{2}}_{\tilde{s}+1, \alpha }\|{\omega}\|^{1-\delta}_{\tilde{s}, \alpha } \|{\omega}\|^{\delta}_{\tilde{s}+1, \alpha }\\ \nonumber
		&=c_{\tilde{s}} \alpha  \|{\omega}\|^{{\tilde{s}}+\frac{1}{2}}_{\tilde{s}, \alpha } \|{\omega}\|^{\frac{5}{2}-s}_{\tilde{s}+1, \alpha }.
	\end{align*}
	Case (b): $-\frac{1}{2}<\tilde{s}<0$.\\
	Case (b1): if $|j |\leq |k |$, then $\left | {\eta}\right |\geq |j |$, we have:
	$\displaystyle
		\left | {\eta}\right |^{{\tilde{s}}}\leq  |j |^{{\tilde{s}}}.
	$ Therefore,
	\begin{align*}
		P_2 & \leq \alpha \sum_{l+j=k} |k | |\hat{u}_{l}| |\hat{\omega}_{j}| |\hat{\omega}_{k}||k |^{\tilde{s}} e^{\alpha |k |} |j |^{{\tilde{s}}} e^{\alpha |j |}e^{\alpha |l |} |l |\ \\
		&\leq \alpha  \sum_{l+j=k}|k |^{1-\delta}|j |^{{\tilde{s}}}  (|l ||\hat{u}_{l}|e^{\alpha |l |})\cdot (|\hat{\omega}_{j}|e^{\alpha |j |})\cdot (|\hat{\omega}_{k}||k |^{{\tilde{s}}+\delta} e^{\alpha |k |} ) \\
		&\leq  \alpha  	\|{\omega}_1\ast {\omega}_3\|_{\dot{H}^{1-\delta}} \|{\omega}\|_{\tilde{s}+\delta,  \alpha }.
	\end{align*}
	
	When $- \frac{1}{2}< \tilde{s}< 0$ with $0<\delta<1$, from Lemma \ref{lem-gev1} with $s_1=\frac{5+2\tilde{s}-2\delta}{4}$ and $s_2=\frac{5-2\delta-2\tilde{s}}{4}$, we have
	\begin{align*}
		\|{\omega}_1\ast {\omega}_3\|_{\dot{H}^{1-\delta}}\leq c_{\tilde{s}}\|{\omega}_1\|_\frac{5+2\tilde{s}-2\delta}{4} \|{\omega}_3\|_\frac{5-2\delta-2\tilde{s}}{4} =c\|{\omega}\|^2_{\frac{5+2\tilde{s}-2\delta}{4}, \alpha }.
	\end{align*}
	
	Therefore,
	$\displaystyle
		P_2\leq c_{\tilde{s}} \alpha  \|{\omega}\|^2_{\frac{5+2\tilde{s}-2\delta}{4}, \alpha } \|{\omega}\|_{\tilde{s}+\delta,\alpha }.
$
	
	When $- \frac{1}{2}< \tilde{s}< 0$ with $\frac{1}{2}-\tilde{s}<\delta<1$, we have
	\begin{align*}
		\|{\omega}\|^2_{\frac{5+2\tilde{s}-2\delta}{4}, \alpha }\leq c_{\tilde{s}} \|{\omega}\|^{\frac{2\delta+2\tilde{s}-1}{2}}_{\tilde{s}, \alpha } \|{\omega}\|^{\frac{5-2\delta-2\tilde{s}}{2}}_{\tilde{s}+1, \alpha }\ 
	\text{and}\ 
		\|{\omega}\|_{\tilde{s}+\delta, \alpha }\leq c_{\tilde{s}} \|{\omega}\|^{1-\delta}_{\tilde{s}, \alpha } \|{\omega}\|^{\delta}_{\tilde{s}+1, \alpha },
	\end{align*}
	we have
	\begin{align*}
		P_2&\leq c_{\tilde{s}} \alpha \|{\omega}\|^{\frac{2\delta+2\tilde{s}-1}{2}}_{\tilde{s}, \alpha } \|{\omega}\|^{\frac{5-2\delta-2\tilde{s}}{2}}_{\tilde{s}+1, \alpha }\|{\omega}\|^{1-\delta}_{\tilde{s}, \alpha } \|{\omega}\|^{\delta}_{\tilde{s}+1, \alpha }\\ \nonumber
		&=c_{\tilde{s}} \alpha  \|{\omega}\|^{{\tilde{s}}+\frac{1}{2}}_{\tilde{s}, \alpha } \|{\omega}\|^{\frac{5}{2}-s}_{\tilde{s}+1, \alpha }.
	\end{align*}
	Case (b2): if $|j |> |k |$, then $\left | {\eta}\right |\geq |k |$, we have:
$\displaystyle
		\left | {\eta}\right |^{{\tilde{s}}}\leq  |k |^{{\tilde{s}}}.
$
	
	Therefore
	\begin{align*}
		P_2 & \leq \alpha \sum_{l+j=k} |k | |\hat{u}_{l}| |\hat{\omega}_{j}| |\hat{\omega}_{k}||k |^{\tilde{s}} e^{\alpha |k |} |k |^{{\tilde{s}}} e^{\alpha |j |}e^{\alpha |l |} |l |\ \\
		&\leq \alpha  \sum_{l+j=k}|k |^{{\tilde{s}}+1-\delta} (|l ||\hat{u}_{l}|e^{\alpha |l |})\cdot (|\hat{\omega}_{j}|e^{\alpha |j |})\cdot (|\hat{\omega}_{k}||k |^{{\tilde{s}}+\delta} e^{\alpha |k |} ) \\
		&\leq  \alpha  	\|{\omega}_1\ast {\omega}_1\|_{\dot{H}^{{\tilde{s}}+1-\delta}} \|{\omega}\|_{\tilde{s}+\delta,  \alpha }.
	\end{align*}
	
	When $- \frac{1}{2}< \tilde{s}< 0$ with $0<\delta<1$, from Lemma \ref{lem-gev1} with $s_1=s_2=\frac{5+2\tilde{s}-2\delta}{4}$, we have
	\begin{align*}
		\|{\omega}_1\ast {\omega}_2\|_{\dot{H}^{{\tilde{s}}+1-\delta}} \leq c_{\tilde{s}}\|{\omega}\|^2_{\frac{5+2\tilde{s}-2\delta}{4}, \alpha }.
	\end{align*}
	
	Therefore,
$\displaystyle
		P_2\leq c_{\tilde{s}} \alpha  \|{\omega}\|^2_{\frac{5+2\tilde{s}-2\delta}{4}, \alpha } \|{\omega}\|_{\tilde{s}+\delta,\alpha }.
$
	
	When $- \frac{1}{2}< \tilde{s}< 0$ with $\frac{1}{2}-\tilde{s}<\delta<1$, we have
	\begin{align*}
		\|{\omega}\|^2_{\frac{5+2\tilde{s}-2\delta}{4}, \alpha }\leq c_{\tilde{s}} \|{\omega}\|^{\frac{2\delta+2\tilde{s}-1}{2}}_{\tilde{s}, \alpha } \|{\omega}\|^{\frac{5-2\delta-2\tilde{s}}{2}}_{\tilde{s}+1, \alpha }\
	\text{and}\
		\|{\omega}\|_{\tilde{s}+\delta, \alpha }\leq c_{\tilde{s}} \|{\omega}\|^{1-\delta}_{\tilde{s}, \alpha } \|{\omega}\|^{\delta}_{\tilde{s}+1, \alpha }.
	\end{align*}
	
	Therefore
	\begin{align*}
		P_2&\leq c_{\tilde{s}} \alpha \|{\omega}\|^{\frac{2\delta+2\tilde{s}-1}{2}}_{\tilde{s}, \alpha } \|{\omega}\|^{\frac{5-2\delta-2\tilde{s}}{2}}_{\tilde{s}+1, \alpha }\|{\omega}\|^{1-\delta}_{\tilde{s}, \alpha } \|{\omega}\|^{\delta}_{\tilde{s}+1, \alpha }\\ \nonumber
		&=c_{\tilde{s}} \alpha  \|{\omega}\|^{{\tilde{s}}+\frac{1}{2}}_{\tilde{s}, \alpha } \|{\omega}\|^{\frac{5}{2}-s}_{\tilde{s}+1, \alpha }.
	\end{align*}
	
	Combing Case (a) and Case (b), we have
	\begin{align}
		\label{P2}
		P_2\leq
		c_{\tilde{s}} \alpha  \|{\omega}\|^{{\tilde{s}}+\frac{1}{2}}_{\tilde{s}, \alpha } \|{\omega}\|^{\frac{5}{2}-s}_{\tilde{s}+1, \alpha }.
	\end{align}
	
	Combing (\ref{p1case3}) and (\ref{P2}), when $-\frac{1}{2}<\tilde{s}<\frac{3}{2}$, it yields that
	\begin{align*}
		\left|	\left (B(u, {\omega}), A^{{\tilde{s}}}e^{2\alpha A^{\frac{1}{2}}}{\omega} \right ) \right|=P=P_1+P_2\leq c_{\tilde{s}} \|{\omega}\|^{{\tilde{s}}+\frac{3}{2}}_{\tilde{s}, \alpha } \|{\omega}\|^{\frac{3}{2}-\tilde{s}}_{\tilde{s}+1, \alpha }+c_{\tilde{s}} \alpha  \|{\omega}\|^{{\tilde{s}}+\frac{1}{2}}_{\tilde{s}, \alpha } \|{\omega}\|^{\frac{5}{2}-\tilde{s}}_{\tilde{s}+1, \alpha }. 
	\end{align*}
\end{proof}

{\bf Proof of Lemma \ref{lemmaonX}.}
\begin{proof}
	Comparing the terms on the right hand side of (\ref{mainequ1X}), we can expect that there is a region (when $t$ and $X$ are both small), $c_{\tilde{s}}X^{1+ \frac{4}{1+2\tilde{s}}}$ is the dominating term among the two terms on the right hand side. 
	
	In order to find this specific region, we compare $c_{\tilde{s}}X^{1+ \frac{4}{1+2\tilde{s}}}$ with $c_{\tilde{s}}(\beta t)^{\frac{4}{2\tilde{s}-1}}X^{1+  \frac{4}{2\tilde{s}-1}}$.
	
	 If  $c_{\tilde{s}}X^{1+ \frac{4}{1+2\tilde{s}}}\geq c_{\tilde{s}}(\beta t)^{\frac{4}{2\tilde{s}-1}}X^{1+  \frac{4}{2\tilde{s}-1}},$ then
	 $\displaystyle X \leq \frac{c_{\tilde{s}} }{(\beta t)^{\frac{2\tilde{s}+1}{2}}}.$ 
	
	Considering the function
	\begin{align*}
	K(t)=X(t)-\frac{c_{\tilde{s}} }{(\beta t)^{\frac{2\tilde{s}+1}{2}}}.
	\end{align*}

	 From (\ref{mainequ1X}), we observe that $X$ starts with positve initial data and is an increasing function. Moreover, since $X \nearrow \infty$ as $t\nearrow T_{X}$, it will intersect the curve $\displaystyle \frac{c_{\tilde{s}} }{(\beta t)^{\frac{2\tilde{s}+1}{2}}}$. Therefore, there exists a $t_{X}$ such that $K(t_{X})=0$ and $K(t)<0$ when $t<t_{X}$.
	Therefore, when $0<t<t_{X}$ and we have
	\begin{align}
	\label{innerpro5}
	\frac{d X}{dt}<2c_{\tilde{s}}X^{1+ \frac{4}{1+2\tilde{s}}}:=c_{\tilde{s}}X^{1+ \frac{4}{1+2\tilde{s}}}.
	\end{align}
	
	When $0<t<t_{X}$, we compare $X(t)$ with $\varphi(t)$, where, $\varphi(t)$ is the solution of
	\begin{align}
	\label{innerpro6}
	\frac{d \varphi}{dt}=c_{\tilde{s}}\varphi^{1+ \frac{4}{1+2\tilde{s}}},
	\end{align}
	with $\varphi(0)=X(0)$ and $T_{\varphi}$ is the local existence time of $\varphi$.\\
	
	Applying Lemma \ref{lemma:nonlinear_Gronwall} on (\ref{innerpro5}) and (\ref{innerpro6}), we have 
	\begin{align*}
	X(t) < \varphi(t), \text{for\ all}\ t\in \left[0, \min \left\{t_{X}, T_{X}, T_{\varphi}\right\}\right].
	\end{align*}
	
	From (\ref{innerpro6}), $\varphi(t)$ will also intercepts with the curve $\displaystyle \frac{c_{\tilde{s}} }{(\beta t)^{\frac{2\tilde{s}+1}{2}}}$. Denote the interception point as $t_{\varphi}$, then $t_{\varphi}< t_{X}<T_{X}$.
	To calculate $t_{\varphi}$, we have
	\begin{align}
	\label{tphi-n}
	\varphi(t_{\varphi})=\frac{c_{\tilde{s}} }{(\beta t_{\varphi})^{\frac{2\tilde{s}+1}{2}}}.
	\end{align}
	
	Solving (\ref{innerpro6}), we have
	\begin{align}
	\label{tsol-n}
	\varphi(t)=(\varphi(0)^{-\frac{4}{1+2\tilde{s}}}-c_{\tilde{s}} t)^{-\frac{1+2\tilde{s}}{4}}.
	\end{align}
	
	Therefore:
	$\displaystyle
	(\varphi(0)^{-\frac{4}{1+2\tilde{s}}}-c_{\tilde{s}} t_{\varphi})^{-\frac{1+2\tilde{s}}{4}}=\frac{c_{\tilde{s}} }{(\beta t_{\varphi})^{\frac{2\tilde{s}+1}{2}}}.
$
	
	After simplification, we obtain:
$\displaystyle
	\varphi(0)^{-\frac{4}{1+2\tilde{s}}}-c_{\tilde{s}} t_{\varphi}=c_{\tilde{s}} \beta^{2} t_{\varphi}^{2}.
$ Therefore
	\begin{align}
	\label{equt-n}
c_{\tilde{s}} \beta^{2} t_{\varphi}^{2}+c_{\tilde{s}} t_{\varphi}=\varphi(0)^{-\frac{4}{1+2\tilde{s}}}.
	\end{align}
	
	Case (i): when $\displaystyle X({0})\geq \frac{c_{\tilde{s}} }{(\beta)^{\frac{2\tilde{s}+1}{2}}} $, then, $\displaystyle \varphi({0})\geq \frac{c_{\tilde{s}} }{(\beta)^{\frac{2\tilde{s}+1}{2}}}\Rightarrow \varphi({1})> \frac{c_{\tilde{s}} }{(\beta)^{\frac{2\tilde{s}+1}{2}}}$. This implies $t_{\varphi}<1$, then ${t^2_{\varphi}}<t_{\varphi}$, since $\beta<\frac{1}{2}$, we have 
	\begin{align}
	\label{equtphi-n}
	\varphi(0)^{-\frac{4}{1+2\tilde{s}}}\leq c_{\tilde{s}}t_{\varphi},
	\end{align}
	this implies:
	$\displaystyle
	t_{\varphi}\geq \frac{c_{\tilde{s}}}{\varphi(0)^{\frac{4}{1+2\tilde{s}}}}.
	$ Therefore
	\begin{align}
	\label{TX1}
T_{X}>t_{\varphi}\geq \frac{c_{\tilde{s}}}{\varphi(0)^{\frac{4}{1+2\tilde{s}}}}=\frac{c_{\tilde{s}}}{X(0)^{\frac{4}{1+2\tilde{s}}}}.
	\end{align}
	 
	 Case (ii): when $\displaystyle X({0})< \frac{c_{\tilde{s}} }{(\beta)^{\frac{2\tilde{s}+1}{2}}}$, then, $\displaystyle \varphi({0})< \frac{c_{\tilde{s}} }{(\beta)^{\frac{2\tilde{s}+1}{2}}}$. If $ \varphi({1})> \frac{c_{\tilde{s}} }{(\beta)^{\frac{2\tilde{s}+1}{2}}}$. This implies $t_{\varphi}<1$, same as Case (i), we have:
	 	$\displaystyle
	 T_{X}>\frac{c_{\tilde{s}}}{X(0)^{\frac{4}{1+2\tilde{s}}}}.
	 $
	 
	 If $ \varphi({1})\leq \frac{c_{\tilde{s}} }{(\beta)^{\frac{2\tilde{s}+1}{2}}}$. This implies $t_{\varphi}\geq1$, then ${t^2_{\varphi}}\geq t_{\varphi}$, then (\ref{equt-n}) becomes
	\begin{align}
	\label{equtphi-n-n}
	\varphi(0)^{-\frac{4}{1+2\tilde{s}}}\leq c_{\tilde{s}}t_{\varphi}^2,
	\end{align}
	this implies $\displaystyle
	t_{\varphi}\geq \frac{c_{\tilde{s}}}{\varphi(0)^{\frac{2}{1+2\tilde{s}}}}.
	$ Therefore
	\begin{align}
	\label{TX1n}
	T_{X}>t_{\varphi}\geq \frac{c_{\tilde{s}}}{\varphi(0)^{\frac{2}{1+2\tilde{s}}}}=\frac{c_{\tilde{s}}}{X(0)^{\frac{2}{1+2\tilde{s}}}}.
	\end{align}
	Therefore, in Case (ii), we have
	\begin{align*}
T_{X}>\min\left\{Q, Q^{1/2}\right\},
	\end{align*}
	where $\displaystyle Q=\frac{c_{\tilde{s}}}{X(0)^{\frac{4}{1+2\tilde{s}}}}$.
	\end{proof}
\section*{Acknowledgement} 
A. Biswas and J. Hudson are partially supported by NSF grant DMS-1517027. J. Tian is partially supported by the AMS Simons Travel Grant.


\section*{References}
\begin{thebibliography}{}
	
	\bibitem{ang}
	\newblock S. B. Angenent,   
	\newblock {Nonlinear analytic semiflows},
	\newblock Proceedings of the Royal Society of Edinburgh. Section A. Mathematics.  {115} (1990), no.1--2, 91--107. 
	
	\bibitem{benameur2010blow}
	J. {Benameur},
	{On the blow-up criterion of 3D Navier--Stokes equations},
	{Journal of Mathematical Analysis and Applications},
	{371},
	{719--727},
	{2010}.
	
	\bibitem{benameur2014exponential}
	J. {Benameur},
	{On the exponential type explosion of Navier--Stokes equations},
	{Nonlinear Analysis: Theory, Methods \& Applications},
	{103},
	{87--97},
	{2014}.
	
	\bibitem{benameur2016blow}
	J. {Benameur},
	L. {Jlali},
	{On the Blow-up Criterion of 3D-NSE in Sobolev--Gevrey Spaces},
	{Journal of Mathematical Fluid Mechanics},
	{18},
	{805--822},
	{2016}.
	
	\bibitem{BiSQG}
	A. {Biswas},
	{Gevrey regularity for the supercritical 
		quasi-geostrophic equation},
	{J. Differential Equations},
	{257}, {no. 6}, {1753--1772},
	{2014}.
	
	
	\bibitem{biswas2012}
	\newblock A. Biswas, 
	\newblock {Gevrey regularity for a class of dissipative equations with applications to decay},
	\newblock Journal of Differential Equations {253} (2012), no. 10, 2739--2764.
	
	\bibitem{BiF}
	A. {Biswas}, C. {Foias},
	{On the maximal space analyticity radius for the 3D Navier-Stokes equations and energy cascades},
	{ Ann. Mat. Pura Appl. (4)},
	{193},
	{no. 3}, {739--777},
	{2014}. 
	
	\bibitem{BJMT}
	A. {Biswas}, M. S. {Jolly}, V. R. {Martinez},  E. S. {Titi},
	{Dissipation length scale estimates for turbulent flows: 
		a Wiener algebra approach},
	{J. Nonlinear Sci.},
	{24},
	{441--471},
	{2014}.
	
	\bibitem{BMSSQG}
	A. {Biswas}, V. R. {Martinez},
	P. {Silva},
	{On Gevrey regularity of the supercritical SQG equation 
		in critical Besov spaces},
	{ J. Funct. Anal.}, {269}, 
	{no. 10}, {3083--3119}, {2015}. 
	
	
	
	\bibitem{BiSw}
	A. {Biswas},
	D. {Swanson},
	{Gevrey regularity of solutions to the 3-D Navier-Stokes equations with weighted $l_p$ initial data},
     {Indiana Univ. Math. J.},
      {\bf 56}, {(2007)}, {no.3,} 
      {1157--1188}. 
	
	
	
	
	\bibitem{biswas2010navier}
	A. {Biswas},
	D. {Swanson},
	{Navier--Stokes equations and weighted convolution inequalities in groups},
	{Communications in Partial Differential Equations},
	{35},
	{559--589},
	{2010}.
	
	
	
	\bibitem{biswas2007existence}
	\newblock A. Biswas, D. Swanson,
	\newblock {Existence and generalized Gevrey regularity of solutions to the Kuramoto-Sivashinsky equation in ${\mathbb R}^n$},
	\newblock Journal of Differential Equations {240} (2007), no. 1, 145--163.
	
	\bibitem{cao1999navier}
	\newblock C. Cao, M.A. Rammaha, E. Titi, 
	\newblock {The Navier-Stokes equations on the rotating 2-D sphere: Gevrey regularity and asymptotic degrees of freedom},
	\newblock Zeitschrift f{\"u}r Angewandte Mathematik und Physik (ZAMP) {50} (1999), no.3, 341--360.
	
	\bibitem{bradshaw}
	Z. {Bradshaw},
	{Geometric measure-type regularity criteria for the 3D magnetohydrodynamical system},
	{ Nonlinear Anal.},
	{75}, {no. 16},
	{6180--6190}, {2012}.
	
	\bibitem{BGK}
	Z. {Bradshaw},
	Z. {Gruji\'{c}},
	I. {Kukavica},
	{Local analyticity radii of solutions to the 3D Navier-Stokes equations with locally analytic forcing},
	{ J. Differential Equations},
	{259}, {no. 8}, {3955--3975}, 
	{2015}.
	
	
	\bibitem{cheskidov2016lower}
	A. {Cheskidov},
	K. {Zaya},
	{Lower bounds of potential blow-up solutions of the three-dimensional Navier--Stokes equations in $\dot{H}^\frac{3}{2}$},
	{Journal of Mathematical Physics},
	{57},
	{023101},
	{2016}.
	
	\bibitem{cf}
	P. {Constantin},
	C. {Foias},
	{Navier-stokes equations},
	{University of Chicago Press},
	{1988}.
	
	\bibitem{CM2018}
	J. {Cortissoz},
	J. {Montero},
	{Lower Bounds for Possible Singular Solutions for the Navier–Stokes and EulerEquations Revisited},
	{J. Math. Fluid Mech.},
	{20}, {1--5},
	{2018}.
	
	
	
	\bibitem{cortissoz2014lower}
	J. {Cortissoz},
	J. {Montero},
	C. {Pinilla}
	{On lower bounds for possible blow-up solutions to the periodic Navier--Stokes equation},
	{Journal of Mathematical Physics},
	{55},
	{033101},
	{2014}.
	
	
	\bibitem{doel-t}
	\newblock A. Doelman, E. Titi,
	\newblock {Regularity of solutions and the convergence of the Galerkin method in the complex Ginzburg--Landau equation},
	\newblock Numer. Func. Opt. Anal. {14} (1993), 299--321.
	
	
	\bibitem{dt}
	\newblock C. R. Doering, E. Titi,   
	\newblock {Exponential decay rate of the power spectrum for solutions of the Navier--Stokes equations},
	\newblock Physics of Fluids {7} (1995), 1384--1390.
	
	\bibitem{dong1}
	\newblock H. Dong,  
	\newblock {Dissipative quasi-geostrophic equations in critical Sobolev spaces: smoothing effect and
		global well-posedness},
	\newblock Discrete and Continuous Dynamical Systems {26} (2010), no.4, 1197--1211.
	
	
	
	\bibitem{dong}
	\newblock H. Dong, D. Li,  
	\newblock {Spatial analyticity of the solutions to the subcritical dissipative quasi-geostrophic equations},
	\newblock Arch. Rational Mech. Anal. {189} (2008), no.1, 131--158.
	
	\bibitem{ferrari1998gevrey}
	\newblock A.B. Ferrari, E. Titi, 
	\newblock {Gevrey regularity for nonlinear analytic parabolic equations},
	\newblock Communications in Partial Differential Equations {23} (1998), no.1, 424--448.
	
	\bibitem{foias}
	\newblock C. Foias,   
	\newblock {What do the Navier-Stokes equations tell us about turbulence? Harmonic Analysis and Nonlinear Differential Equations},
	\newblock Contemp. Math. {208} (1997), 151--180.
	
	
	\bibitem{foias1989gevrey}
	C. {Foias},
	R. {Temam},
	{Gevrey class regularity for the solutions of the Navier--Stokes equations},
	{Journal of Functional Analysis},
	{87},
	{359--369},
	{1989}.
	
	\bibitem{fr}
	P. K. {Fritz}, J. C. {Robinson},
	{Parametrising the attractor of the two-dimensional Navier–Stokes equations with a finite number of nodal values},
	{Physica D}, {148},
	{(2001)}, {201--220}.
	
	\bibitem{fkr}
	P. K. {Fritz}, 
	I. {Kukavica},
	J. C. {Robinson},
	{Nodal parametrisation of analytic attractors},
	{Disc. Cont. Dyn. Syst.}, {7}, {(2001)}, 
	{643--657}.
	
	\bibitem{pav}
	\newblock P. Germain, N. Pavlovi{\'c}, G. Staffilani,   
	\newblock {Regularity of Solutions to the Navier-Stokes Equations Evolving from Small Data in $BMO^{-1}$},
	\newblock International Mathematics Research Notices {21} (2007).
	
	\bibitem{giga1986solutions}
	Y. {Giga, Y},
	{Solutions for semilinear parabolic equations in $L^p$ and regularity of weak solutions of the Navier--Stokes system},
	{Journal of Differential Equations},
	{62},
	{186--212},
	{1986}.
	
	\bibitem{g}
	\newblock Z. Grujic,
	\newblock {The geometric structure of the super-level sets and regularity for 3D Navier-Stokes equations},
	\newblock Indiana Univ. Math. J. {50} (2001), no.3, 1309--1317.
	
	\bibitem{gk}
	\newblock Z. Grujic, I. Kukavica, 
	\newblock {Space Analyticity for the Navier-Stokes and Related Equations with Initial Data in $L_p$},
	\newblock Journal of Functional Analysis {152} (1998), no.2, 447--466.
	
	
	
	\bibitem{gk2}
	\newblock Z. Grujic, I. Kukavica, 
	\newblock {Space analyticity for the nonlinear heat equation in a bounded domain},
	\newblock Journal of Differential Equations {154} (1999), no.1, 42--54. 
	
	\bibitem{hkr}
	\newblock W.D. Henshaw, H.O. Kreiss, L.G. Reyna, 
	\newblock {Smallest scale estimates for the Navier-Stokes equations for incompressible fluids},
	\newblock Arch. Rational Mech. Anal. {112} (1990), no.1, 21--44.
	
	
	\bibitem{hkr1}
	\newblock W.D. Henshaw, H.O. Kreiss, L.G. Reyna, 
	\newblock {On smallest scale estimates and a comparison of the vortex method to the pseudo-spectral method},
	\newblock Vortex Dynamics and Vortex Methods (1990), 303--325.
	
	\bibitem{online1}
	J. {Hunter},
	{Nonlinear Evolution Equations},
	{https://www.math.ucdavis.edu/~hunter/notes/nonlinev.pdf}.
	
	\bibitem{Kalantarov}
	\newblock V. Kalantarov, B. Levant, E. Titi, 
	\newblock {Gevrey regularity for the attractor of the 3D Navier-Stoke-Voight equations},
	\newblock  J. Nonlinear Sci. {19} (2009), no. 2, 133--152.
	
	\bibitem{kuk1}
	\newblock I. Kukavica,
	\newblock {Level sets of the vorticity and the stream function for the 2-D periodic Navier-Stokes equations with
		potential forces},
	\newblock Journal of Differential Equations {126} (1996), 374--388.
	
	
	
	
	
	\bibitem{kuk}
	\newblock I. Kukavica,
	\newblock {On the dissipative scale for the Navier-Stokes equation},
	\newblock Indiana Univ. Math. J. {48} (1999), no.3, 1057--1082.
	
	\bibitem{leray1934mouvement}
	J. {Leray},
	{Sur le mouvement d'un liquide visqueux emplissant l'espace},
	{Acta mathematica},
	{63},
	{193--248},
	{1934}.
	
	\bibitem{lorenz2017properties}
	J. {Lorenz},
	P. {Zingano},
	{Properties at potential blow-up times for the incompressible Navier--Stokes equations},
	{Bol. Soc. Parana. Mat},
	{5},
	{127--158},
	{2017}.
	
	\bibitem{Ly}
	\newblock H.V. Ly, E. Titi,  
	\newblock {Global gevrey regularity for the B\'enard convection in porous medium with zero Darcy-Prandtl number},
	\newblock J. Nonlinear Sci. {9} (1999), no. 3, 333--362.
	
	\bibitem{masuda}
	\newblock K. Masuda,  
	\newblock {On the analyticity and the unique continuation theorem for Navier-Stokes equations},
	\newblock Proc. Japan Acad. Ser. A Math. Sci. {43} (1967), 827--832.
	
	\bibitem{mccormick2016lower}
	D. {McCormick},
	E. {Olson},
	J. {Robinson},
	J. {Rodrigo},
	A. {Vidal-L{\'o}pez},
	Y. {Zhou},
	{Lower Bounds on Blowing-Up Solutions of the Three-Dimensional Navier--Stokes Equations in $\dot{H}^{\frac{3}{2}}$, $\dot{H}^{\frac{5}{2}}$, and $\dot{B}^{\frac{5}{2}}_{2,1}$},
	{SIAM Journal on Mathematical Analysis},
	{48},
	{2119--2132},
	{2016}.
	
	\bibitem{miura2006}
	\newblock H. Miura, O. Sawada,   
	\newblock {On the regularizing rate estimates of Koch-Tataru's solution to the Navier-Stokes equations},
	\newblock Asymptotic Analysis.  {49} (2006), no.1, 1--15. 
	
	\bibitem{oliver2000remark}
	\newblock M. Oliver, E. Titi,
	\newblock {Remark on the Rate of Decay of Higher Order Derivatives for Solutions to the Navier-Stokes Equations in ${\mathbb R}^n$},
	\newblock Journal of Functional Analysis {172} (2000), no.1, 1--18.
	
	
	
	\bibitem{optimaltiti}
	\newblock M. Oliver, E. Titi,  
	\newblock {On the domain of analyticity for solutions of second order analytic nonlinear differential equations},
	\newblock Journal of Differential Equations {174} (2001), no.1, 55--74.
	
	\bibitem{robinson2014local}
	J. {Robinson},
	W. {Sadowski},
	{A local smoothness criterion for solutions of the 3D Navier--Stokes equations},
	{Rendiconti del Seminario Matematico della Universit{\`a} di Padova},
	{131},
	{159--178},
	{2014}.
	
	\bibitem{robinson2012lower}
	J. {Robinson},
	W. {Sadowski},
	R. {Silva},
	{Lower bounds on blow-up solutions of the three-dimensional Navier--Stokes equations in homogeneous Sobolev spaces},
	{Journal of Mathematical Physics},
	{53},
	{115618},
	{2012}.
	
	\bibitem{swanson2011gevrey}
	D. {Swanson},
	{Gevrey regularity of certain solutions to the Cahn-Hilliard equation with rough initial data},
	{Methods and Applications of Analysis},
	{18},
	{417--426},
	{2011}.
	\bibitem{temam1995navier}
	R. {Temam},
	{Navier-Stokes equations and nonlinear functional analysis},
	{66},
	{1995}.
	
\end{thebibliography}
\end{document}